\documentclass[10pt,a4paper]{article}
\RequirePackage[numbers]{natbib}
\RequirePackage[colorlinks,citecolor=blue,urlcolor=blue]{hyperref}
\RequirePackage{graphicx}
\usepackage{color}
\usepackage[margin=1in]{geometry}
\usepackage{amsmath,amsthm}
\usepackage{amssymb}
\usepackage[english]{babel}
\usepackage[utf8]{inputenc}
\usepackage[caption=false]{subfig}
\usepackage{enumerate}
\usepackage{mathrsfs}
\usepackage{bm}
\usepackage{floatpag}
\usepackage{setspace}
\theoremstyle{plain}

\newtheorem{theorem}{Theorem}[section]
\newtheorem{lemma}[theorem]{Lemma}
\theoremstyle{definition}

\newtheorem{remark}{Remark}
\theoremstyle{remark}

\newcommand{\bR}{\mathbb{R}}
\newcommand{\vv}[1]{\mathbf{#1}} % bold vectors
\newcommand{\eps}{\varepsilon}

\providecommand{\keywords}[1]{\textbf{\textit{Keywords---}} #1}
\title{An MM-type algorithm for estimation of semiparametric finite mixture models with a copula-based dependence structure}
\author{Michael Levine} % 
\date{%
Department of Statistics, Purdue Univesity, 150 N. University St., West Lafayette, IN 47907, USA\\\texttt{mlevins@purdue.edu}}
\begin{document}
\maketitle
%%%%%%%%%%%%%%%%%%%%%%%%%%%%%%%%%%%%%%%%%%%%%%%
%% Only one address is permitted per author. %%
%% Only division, organization and e-mail is %%
%% included in the address.                  %%
%% Additional information can be included in %%
%% the Acknowledgments section if necessary. %%
%% ORCID can be inserted by command:         %%
%% \orcid{0000-0000-0000-0000}               %%
%%%%%%%%%%%%%%%%%%%%%%%%%%%%%%%%%%%%%%%%%%%%%%%
\begin{abstract}
In this manuscript, we consider a finite nonparametric mixture model with non-independent marginal density functions. Dependence between the marginal densities is modeled using a copula device. Until recently, no deterministic algorithms capable of estimating components of such a model have been available. A deterministic algorithm that is capable of this has been proposed by this author earlier. That algorithm seeks to maximize a smoothed nonparametric penalized log-likelihood; it seems to perform well in practice but does not possess the monotonicity property. In this manuscript, we introduce a deterministic MM (Minorization-Maximization) algorithm for estimation of components of this model that is also maximizing a smoothed penalized nonparametric log-likelihood but that is monotonic with respect to this objective functional. Besides the convergence of the objective functional, results on the convergence of sequences of density functions generated by this algorithm are also established. The behavior of this algorithm is illustrated using both simulated datasets and a real dataset. The results illustrate performance that is at least comparable to the earlier algorithm mentioned above. A discussion of the results and possible future research directions make up the last part of the manuscript. 
\end{abstract}

\section{Introduction}
Consider a general finite density mixture model 
\begin{equation}\label{model}
g(\vv x)=\sum_{j=1}^{m}\lambda_{j}f_j(\vv x)
\end{equation}
where the number of components $m$ is known beforehand. We assume that $\vv x=(x_1,\ldots,x_d)^{'}\in \bR^{d}$. As usual, it is assumed that the weights are positive and that they add up to $1$: $\lambda_j> 0$, $\sum_{j=1}^{m}\lambda_j=1$. We view this model as a nonparametric mixture model where individual components $f_{j}$ do not belong to any parametric family. The number of components is fixed since model selection procedures for nonparametric finite density mixtures are not well developed yet; some of the (very limited) relevant research can be found in \cite{kasahara2014non} and \cite{kwon2021estimation}.

The model \eqref{model} is known to be identifiable under certain conditions if each of the densities $f_j$ is equal to the product of the marginal densities: $f_j(\vv x)=\prod_{k=1}^{d}f_{jk}(x_k)$ for every $j=1,\ldots,m$. That latter assumption is known as the conditional independence assumption. The nonparametric finite mixture model with conditionally independent marginals has been originally introduced in \cite{Hall_Zhou} and its identifiability conditions worked out in \cite{Allman_Matias_Rhodes_2009}. In particular, identifiability conditions obtained in \cite{Allman_Matias_Rhodes_2009} require that the number of dimensions $d\ge 3$. This model has been studied in detail in e.g. \cite{benaglia2009like, Biometrika_2011LevineHunterChauveau, zheng2020nonparametric, Bonhomme_2016}. Note that independence assumption is not always a realistic one, however: it is almost always violated in, e.g. RNA-seq data \cite{rau2015co}. Thus, estimation of finite nonparametric mixtures with non-independent marginals seems desirable from the practical viewpoint. 

In order to accommodate this lack of independence, we will use a copula device to model the dependence structure between marginals. As in \cite{levine2024smoothed}, we assume that every $d$-dimensional multivariate cumulative distribution function corresponding to a component of the mixture model \eqref{model}, is represented as a copula of the corresponding marginal cumulative distribution functions. In other words, let $F_{j1}(x_1),\ldots,F_{jd}(x_d)$ be the marginal cumulative distribution functions of the cumulative distribution function $F_{j}(x_1,\ldots,x_d)$ corresponding to the density $f_{j}(x_1,\ldots,x_d)$. According to the Sklar's theorem \cite{nelsen2007introduction} p $18$, we have $F_{j}(x_1,\ldots,x_d)=C_{j}(F_{j1}(x_1),\ldots,F_{jd}(x_d))$ for a function $C_{j}:[0,1]^{d}\rightarrow [0,1]$ which is called a copula. The copula $C_{j}$ is, effectively, a $d$-dimensional cumulative distribution function with uniform marginal distributions. Differentiating this expression $d$ times, one obtains the representation of the $j$th multivariate density function that we need:
\begin{equation}\label{cmpt}
f_{j}(x_1,\ldots,x_d)=c_{j}(F_{j1}(x_1),\ldots,F_{jd}(x_d))\prod_{k=1}^{d}f_{jk}(x_k)
\end{equation}
where $c_{j}$ is the density of the copula $C_{j}$. In this manuscript, we work under the parametric assumption concerning copulas $c_{j}$: each one of them is viewed as belonging to a {\bf known} parametric family indexed by a parameter $\rho_{j}$. This assumption makes the use of the index $j$ unnecessary for copula densities. The index will only be retained for the copula parameter. 
%For example, in the two-dimensional case this may be an FGM (Farlie-Gumbel-Morgenstern) family indexed by the parameter $\rho\in [-1,1]$. 

Let $\bm\lambda=(\lambda_{1},\ldots,\lambda_{m})^{'}$ and $\bm\rho=(\rho_{1},\ldots,\rho_{m})^{'}$. Then, we can define $\bm\theta=(\bm\lambda^{'},\bm\rho^{'})^{'}$ to be the vector of finite-dimensional parameters to be estimated. We also denote $\vv f$ the vector of  multivariate densities $f_{j}$, $j=1,\ldots,m$. 
%The notation we have introduced so far allows us to restate the definition of the $j$th component of the model \eqref{model} as
%\begin{equation}\label{cmpt}
%f_{j}(\vv x;\rho_{j},\bm\phi)=f_{j}(x_1,\ldots, x_d;\rho_j,\bm\phi)=c(F_{j1}(x_1),\ldots,F_{jd}(x_d);\rho_{j})\prod_{k=1}^{d}f_{jk}(x_k)
%\end{equation}
\eqref{model}-\eqref{cmpt} together define a class of mixture models that we investigate. In general, this model is not identifiable. 

%As a basic example, consider a two-component mixture where $g(\vv x)=\frac{1}{2} f_1(\vv x)+\frac{1}{2}f_2(\vv x)$. Clearly, we can also represent $g(\vv x)$ as 
%\[
%g(\vv x)=\frac{1}{2}\left(\frac{f_1(\vv x)+f_{2}(\vv x)}{2}\right)+\frac{1}{2}%%%\left(\frac{f_1(\vv x)+f_{2}(\vv x)}{2}\right)
%\]
%which is also a legitimate two-component mixture of two identical densities.
Proving identifiability of non- and semiparametric finite density mixtures in general is known to be a  difficult problem. To the best of our knowledge, no results at all exist for nonparametric  finite mixtures of multivariate components where the dependence is modeled using copulas. Moreover, there are no known results even for purely parametric copula-based mixture models, such as those studied in \cite{arakelian2014clustering, kosmidis2016model}. If the main task is a parameter estimation, identifiability is an important problem. In this respect, the simulation section of this manuscript is reassuring in that the true parameters with which the data were simulated could be recovered. Thus, our proposed algorithm seems to belong to a group of seemingly well-behaved algorithms proposed for models where identifiability is not guaranteed see e.g. \cite{levine2024smoothed,zhu2019clustering,mazo2019constraining}. Note also that, from a clustering viewpoint, the identifiability problem is less important. In particular, a number of neural network models that are known to be unidentifiable have been used successfully for clustering purposes. 

Note that, in general, estimation of nonparametric mixture models is a fairly new topic of research. Some of the modern overviews of the subject can be found in e.g. \cite{xiang2019overview} and \cite{hunter2024unsupervised}. In particular, there has been very limited research on estimation of nonparametric mixture models with conditionally non-independent components which is exactly the topic considered in this manuscript. Some of the relevant publications are \cite{mazo2017semiparametric}, \cite{mazo2019constraining}, and also \cite{levine2024smoothed}. Also, some manuscripts proposing methods for models with conditionally independent marginals propose ideas that allow extension of relevant algorithms to the case of blocks of dependent marginals making up the entire set of available marginal distributions. The identifiability issue of one such model was considered already in \cite{Allman_Matias_Rhodes_2009}. Algorithms that consider these kinds of models have also been considered in e.g. \cite{hunt2003mixture} and \cite{chauveau2016nonparametric}. \cite{mazo2017semiparametric} considers a special case of the general nonparametric mixture model that allows for a rather general dependence between marginals while assuming that the marginals themselves belong to a location family; \cite{mazo2019constraining}  extends the previous setting to the case where marginals belong to a location-scale family. The algorithms proposed in both of those manuscripts are stochastic and do not optimize any particular objective function. \cite{levine2024smoothed} suggests an algorithm capable of estimating all of the elements of a nonparametric model, including copula parameters describing the dependence between marginal distribution, without the location-scale assumption on the marginal distributions. The proposed algorithm is also deterministic, unlike the algorithms of \cite{mazo2017semiparametric} and \cite{mazo2019constraining}. However, it is not monotonic with respect to any objective functional. Our main contribution in this manuscript is that we propose a deterministic algorithm that can estimate components of a general nonparametric mixture model with conditionally non-independent marginals, like the one considered in \cite{levine2024smoothed}, but that is also monotonic with respect to the appropriate nonparametric smoothed penalized likelihood. The proposed algorithm also preserves the known copula structure.  
%\cite{du2024full} introduced a procedure that allows joint fitting of a nonparametric mixture, selecting the number of components, and selecting the variables important for clustering. The latter method, however, has a significant shortcoming in that its density estimation method may not be desirable in practice. In particular, the authors themselves admit that ``...we advise to use the proposed approach only for model estimation. When the model has been selected, we suggest to use a kernel-based method for density estimation. Indeed, the bin-density estimates are known to be outperformed by kernel-based estimators. Thus, for a real data analysis,
%we advise to use the proposed approach for model selection then, for the selected model, to perform density estimation with a EM-like algorithm %\cite{benaglia2009like} or by maximizing the smoothed log-likelihood \cite{Biometrika_2011LevineHunterChauveau}". A general overview of these developments can be found in \cite{hunter2024unsupervised}. 

The rest of the manuscript is structured as follows. Section $2$ introduces the algorithm we propose that can estimate a general multivariate nonparametric finite mixture model with the dependence between marginals modeled using copulas. We show that the proposed algorithm is monotonic with respect to a certain nonparametric smoothed likelihood. Section $3$ provides some additional results concerning convergence of the functional sequences defined by the proposed algorithm. Section $4$ analyzes performance of our algorithm using some simulation studies. Section $5$ presents an application to a real dataset. Finally, the conclusion section provides a discussion of obtained results and suggests some possible directions for future research.

\section{Algorithm}\label{algo}

\subsection{General algorithm}
In this subsection, we only consider the general model \eqref{model} without explicit specification of the marginal dependence structure. We begin with introducing the objective functional of our algorithm. First, we define a subset of a linear vector function space 
\[
{\cal F}=\{\vv f=(f_1,\ldots,f_{m})^{'}:0<f_j\in L_{1}(\Omega),\log f_{j}\in L_{1}(\Omega),j=1,\ldots,m\}
\]
where $\Omega\subset \bR^{d}$ is a subset of the $d$-dimensional Euclidean space $\bR^{d}$. To proceed, we take $K(\vv u)$ to denote some kernel density function in the space $\bR^{d}$, that is we assume that $K(\vv u)\geq 0$ and $\int K(\vv u)\,d\vv u=1$. Let $H$ be a positive definite symmetric $d\times d$ bandwidth matrix. In this manuscript, we use $|H|$ to denote the determinant of a generic matrix $H$. Then, a rescaled version of the kernel density function is $K_{H}(\vv u):=|H|^{-1/2}K(H^{-1/2}\vv u)$. Also, we define a linear smoothing operator ${\cal S}$ for any function $f$ as 
\[
{\cal S}f(\vv x)=\int_{\Omega}K_{H}(\vv x-\vv u)f(\vv u)\,d\vv u.
\]
For our purposes, we consider a situation where a smoothing kernel used for smoothing $f_i$ in the definition of the operator ${\cal S}$ may be different from the one used to smooth $f_j$ when $i\ne j$. In such a case, the kernel will acquire an additional subscript $j$ as in $K_{j,H}$. To reflect this situation, we also define an extended linear operator ${\cal S}\vv f=({\cal S}_{1}f_1,\ldots,{\cal S}_{m}f_m)^{'}$. We also define a nonlinear smoothing vector-valued operator ${\cal N}$ as 
\[
{\cal N}\vv f(\vv x)=\left({\cal N}_{1}f_{1}(\vv x),\ldots,{\cal N}_{m}f_{m}(\vv x)\right)
\]
where 
\begin{equation}
{\cal N}_{j}f_{j}(\vv x)=\exp\{({\cal} S_{j}\log f_j)(\vv x)\}=\exp\left\{\int_{\Omega}K_{j,H}(\vv x-\vv u)\log f_j(\vv u)\,d\vv u\right\}
\end{equation}
for $j=1,\ldots,m$ and kernel functions potentially depending on $j$. We may note here that the idea of smoothing the logarithm of the density function goes back to \cite{silverman1986density}, where a penalty based on the second derivative of the log-density is discussed. 
To simplify and shorten the notation, we also introduce the finite mixture operators
\[
{\cal M}_{\bm\lambda}\vv f(\vv x)=\sum_{j=1}^{m}\lambda_jf_j(\vv x)
\]
and 
\[
{\cal M}_{\bm\lambda}{\cal N}\vv f(\vv x)=\sum_{j=1}^{m}\lambda_j{\cal N}_{j}f_j(\vv x).
\]
Note that, according to the above notation, ${\cal M}_{\bm\lambda}\vv f(\vv x)=g(\vv x)$.  Now, we are ready to introduce the appropriate objective functional 
\begin{equation}\label{sm_lik}
l(\vv f,\bm\lambda)=\int_{\Omega}g(\vv x)\log \frac{g(\vv x)}{({\cal M}_{\bm\lambda}{\cal N}\vv f)(\vv x)}\,d\vv x.
\end{equation}

In what follows, we would like to clarify why the functional \eqref{sm_lik} makes a sensible choice. We begin with some new notation. Let us use the notation 
\[
D(a|b)=\int \left\{a(\vv x)\log \frac{a(\vv x)}{b(\vv x)}+b(\vv x)-a(\vv x)\right\}\,d\vv x
\] for the generalized Kullback-Leibler divergence between arbitrary non-negative functions $a(\vv x)$ and $b(\vv x)$. As was shown earlier in \cite{Biometrika_2011LevineHunterChauveau}, this functional can be represented as 
\begin{equation}\label{KL_rep}
\mathnormal{l}(\vv f,\bm\lambda)=D(g|{\cal M}_{\bm\lambda}{\cal N}\vv f)+1-\sum_{j=1}^{m}\lambda_j\int {\cal N}_{j}f_j(\vv x)\,d\vv x
\end{equation}
To make the notation more concise, we will omit the subscript $\Omega$ under the integral sign from now on.

By Jensen's inequality, the smoothed ``subdensity" ${\cal M}_{\bm\lambda}{\cal N}\vv f(\vv x)$ has an integral that is less than or equal to $1$, although it is still non-negative.
 Thus, the first additive term in \eqref{KL_rep} is analogous to a regular likelihood function for the ``model" defined by the subdensity $\sum_{j=1}^{m}\lambda_{j}{\cal N}_{j}f_{j}(x_j)$. We place the word ``model" in quotation marks since it does not define the true probabilistic model of the data generating process. The second term can be viewed as a ``penalty" functional that is always non-negative, is equal to zero if no smoothing occurs and is positive for any non-zero bandwidth, again due to Jensen's inequality. In other words, the second term becomes larger when the integral of the ``non-smooth part" of the density function $\sum_{j=1}^{m}\lambda_{j}f_{j}(x_j)$ becomes smaller. Thus, the functional $l(\vv f,\bm\lambda)$ can be viewed as a type of {\it smoothed nonparametric penalized log-likelihood}. We would like to note here that terms ``maximum penalized likelihood estimation" and ``maximum smoothed likelihood estimation" in the nonparametric contexts have been used before e.g. \cite{eggermont2001maximum} pp. 7-8.

Our goal is to find a minimizer of $l(\vv f,\bm\lambda)$ subject to the constraint that $\lambda_{j}> 0$ and $\sum_{j=1}^{m}\lambda_{j}=1$. We achieve this goal by deriving an iterative algorithm that possesses a descent property with respect to the functional $l(\vv f,\bm\lambda)$. In other words, we want to ensure that the value of the functional  does not decrease from one iteration to the next. 

Let us define an iteration operator $G_{\bm\lambda}\vv f(\vv x)=\left\{G_{\bm\lambda,1}f_{1}(\vv x),\ldots,G_{\bm\lambda,m}f_{m}(\vv x)\right\}$ where 
\begin{equation}\label{it_op}
G_{j}f_j(\vv x)=\alpha_{j}\int K_{j,H}(\vv x-\vv u)\frac{g(\vv u){\cal N}_{j}f_j(\vv u)}{{\cal M}_{\bm\lambda}{\cal N}\vv f(\vv u)}\,d\vv u
\end{equation}
where $\alpha_j$ is a proportionality constant that makes $G_{\bm\lambda,j}f_j(\vv x)$ integrate to $1$. 

Let the starting point be $(\vv f^{0},\bm\lambda^{0})$. When using this notation, the individual initial weights are denoted $\{\lambda_{j}^{0}\}_{j=1}^{m}$ and the individual density functions are $\{f_{j}^{0}\}_{j=1}^{m}$. Our goal is to find $(\vv f,\bm\lambda)$ such that the difference $l(\vv f^{0},\bm\lambda^{0})-l(\vv f,\bm\lambda)\ge 0$; if this inequality is true, it would imply the negative smoothed nonparametric likelihood $l(\vv f,\bm\lambda)$ decreases after this iteration. In the same way as for $(\vv f^{0},\bm\lambda^{0})$ earlier, we introduce $\{\lambda_{j}\}_{j=1}^{m}$ and $\{f_{j}\}_{j=1}^{m}$. 
 
 To proceed, let us first denote $w_{j}^{0}(\vv x):=\frac{\lambda_{j}^{0}{\cal N}_{j}f_{j}^{0}(\vv x)}{{\cal M}_{\bm\lambda^{0}}{\cal N}\vv f^{0}(\vv x)}$. Note that $\{w_{j}^{0}(\vv x)\}$ is a sequence of positive weights that adds up to $1$, so $\sum_{j=1}^{m}w_j^{0}(\vv x)=1$. We now prove the lemma establishing monotonicity of the proposed general algorithm. For simplicity, we establish it in the case where $\Omega$ is a compact subset of $\bR^{d}$. 
\begin{lemma}\label{basic_l}
Let $\vv f=\{f_{j}\}_{j=1}^{m}\in {\cal F}$ where $f_j=G_{j}f_{j}^{0}$, $j=1,\ldots,m$. Let also $\bm\lambda=\{\lambda_{j}\}_{j=1}^{m}$ where $\lambda_{j}=\int g(\vv x)w_{j}^{0}(\vv x)\,d\vv x $. Finally, we assume that $\Omega\subset \bR^{d}$ is compact. Then,
\[
l(\vv f^{0},\bm\lambda^{0})-l(\vv f,\bm\lambda)\ge 0.
\]
\end{lemma} 
\begin{proof}
By definition, 
\begin{align*}
&l(\vv f^{0},\bm\lambda^{0})-l(\vv f,\bm\lambda)=\int g(\vv x)\log \frac{{\cal M}_{\bm\lambda}{\cal N}\vv f(\vv x)}{{\cal M}_{\bm\lambda^{0}}{\cal N}\vv f^{0}(\vv x)}\,d\vv x\\
&=\int g(\vv x)\log \sum_{j=1}^{m}\frac{\lambda_{j}^{0}{\cal N}_{j}f_{j}^{0}(\vv x)}{{\cal M}_{\bm\lambda^{0}}{\cal N}\vv f^{0}(\vv x)}\frac{\lambda_{j}{\cal N}_{j}f_{j}(\vv x)}{\lambda_{j}^{0}{\cal N}_{j}f_{j}^{0}(\vv x)}\,d\vv x.
\end{align*}

Thus, using concavity of the logarithm function in conjunction with Jensen's inequality, one finds out that
\[
l(\vv f^{0},\bm\lambda^{0})-l(\vv f,\bm \lambda) \ge \int g(\vv x)\sum_{j=1}^{m}w_{j}^{0}(\vv x)\log \frac{\lambda_{j}{\cal N}_{j}f_{j}(\vv x)}{\lambda_{j}^{0}{\cal N}_{j}f_{j}^{0}(\vv x)}\,d\vv x.
\]
By definition of the nonlinear smoothing operator ${\cal N}$, we have the ratio 
\[
\frac{{\cal N}_{j}f_j(\vv x)}{{\cal N}_{j}f_j^{0}(\vv x)}=\frac{{\cal N}_{j}G_{j}f_j^{0}(\vv x)}{{\cal N}_{j}f_j^{0}(\vv x)}=\exp\left\{\int K_{j,H}(\vv x-\vv u)\log\frac{G_{j}f_j^{0}(\vv u)}{f_j^{0}(\vv u)}\,d\vv u\right\}.
\]
 Splitting the logarithmic function of the product in the above, and using the full expression for weights $w_{j}^{0}$ we find that 
\begin{align}\label{dif1}
&l(\vv f^{0},\bm\lambda^{0})-l(\vv f,\bm\lambda)\ge \sum_{j=1}^{m}\log\frac{\lambda_j}{\lambda_j^{0}}\int g(\vv x)\frac{\lambda_j^{0}{\cal N}_{j}f_j^{0}(\vv x)}{{\cal M}_{\bm\lambda^{0}}{\cal N}\vv f^{0}(\vv x)}\,d\vv x\\
&+\sum_{j=1}^{m}\lambda_j^{0}\int g(\vv x)\frac{{\cal N}_{j}f_j^{0}(\vv x)}{{\cal M}_{\bm\lambda^{0}}{\cal N}\vv f^{0}(\vv x)}\left[\int K_{j,H}(\vv x-\vv u)\log \frac{f_j(\vv u)}{f_j^{0}(\vv u)}\,d\vv u\right]\,d\vv x\nonumber
\end{align}

The second term in the above can be rewritten, using Fubini theorem, as 
\[
\sum_{j=1}^{m}\lambda_j^{0}\int \log\frac{f_j(\vv u)}{f_j^{0}(\vv u)}\left[\int K_{j,H}(\vv x-\vv u)g(\vv x)\frac{{\cal N}_{j}f_j^{0}(\vv x)}{{\cal M}_{\bm\lambda^{0}}{\cal N}\vv f^{0}(\vv x)}\,d\vv x\right]\,d\vv u.
\]
The quantity inside the square brackets is equal to, by \eqref{it_op}, $\frac{f_j(\vv u)}{\alpha_j}$. Therefore, the second term becomes 
\[
\sum_{j=1}^{m}\frac{\lambda_j^{0}}{\alpha_j}\int f_j(\vv u)\log\frac{f_j(\vv u)}{f_j^{0}(\vv u)}\,d\vv u=\sum_{j=1}^{m}\frac{\lambda_j}{\alpha_j}D(f_j|f_j^{0})\ge 0.
\]

Note that the first term on the right-hand side of \eqref{dif1} can be written as $\sum_{j=1}^{m}\log\frac{\lambda_j}{\lambda_j^{0}}\int g(\vv x)w_{j}^{0}(\vv x)\,d\vv x$. It can be maximized with respect to $\lambda_j$, $j=1,\ldots,m$ under the constraint $\sum_{j=1}^{m}\lambda_j=1$ using the standard Lagrange multipliers method. The resulting maximizer is $\lambda_j=\int g(\vv x)w_{j}^{0}(\vv x)\,d\vv x$. Since the choice of $\lambda_j=\lambda_j^{0}$ makes this term equal to zero, it is clear that the value of this term when $\lambda_j=\int g(\vv x)w_{j}^{0}(\vv x)\,d\vv x$, $j=1,\ldots,m$  is greater then or equal to zero. Therefore, indeed, $l(\vv f^{0},\bm\lambda^{0})-l(\vv f,\bm \lambda)\ge 0$.
\end{proof}
\begin{remark}
Note also that, even if the weights $\lambda_j$ are left unchanged between iterations, the difference between the values of smoothed regularized nonparametric likelihood $l(\vv f,\bm\lambda)$ remains non-negative. Indeed, assuming that $\hat{\bm \lambda}=\bm\lambda$ identically, one will be left with only the second term on the right hand side of \eqref{dif1}. Maximizing this term with respect to $\vv f$, one obtains the same result as above.
\end{remark}
\begin{remark}
The algorithm proposed is somewhat similar to that proposed in \cite{Biometrika_2011LevineHunterChauveau} except that it does not operate at the level of marginal densities but rather at the level of multivariate density component functions. In \cite{Biometrika_2011LevineHunterChauveau}, we had to design the algorithm in such a way that the conditional independence of marginal densities is preserved from one iteration to the other since it was an outstanding feature of the model we investigated. No such concern is present when dealing with the model \eqref{model} and so we operate at the level of multivariate density components themselves.
\end{remark}
\begin{remark}
The assumption of compactness of $\Omega\subset \bR^{d}$ has been made for convenience of illustration. If one wants to remove it, some additional assumptions have to be made to guarantee that the integral $\int g(\vv x)\frac{{\cal N}_{j}f_j^{0}(\vv x)}{{\cal M}_{\bm\lambda^{0}}{\cal N}\vv f^{0}(\vv x)}\left[\int K_{j,H}(\vv x-\vv u)\left\vert\log \frac{f_j(\vv u)}{f_j^{0}(\vv u)}\right\vert\,d\vv u\right]\,d\vv x $ is finite which, in turn, implies that the Fubini theorem can be used. Indeed, let us additionally assume that all kernels $K_{j}(\cdot)$ are bounded from above uniformly by some positive constant $D$ and that all of $\lambda_{j}^{0}>0$. Then, using the definition $f_{j}=G_{j}f_{j}^{0}$, the expression  
\begin{align*}
&\log f_j(\vv u)\le\log\left\{\alpha_{j}\int K_{j,H}(\vv u-\vv z)\frac{g(\vv z){\cal N}f_{j}^{0}(\vv z)}{{\cal M}_{\bm\lambda}{\cal N}{\vv f}^{0}(\vv z)}\,dz\right\}\le \log \left(\frac{D\alpha_{j}}{\min_{1\le j \le m}\lambda_{j}^{0}}\right)
\end{align*}
and, therefore, we can guarantee that the integral $\int K_{j,H}(\vv x-\vv u)\log \frac{f_j(\vv u)}{f_j^{0}(\vv u)}\,d\vv u$ is finite as long as 
\[
\max\limits_{1\le j \le M}\int |\log f_{j}^{0}(\vv u)|\,du <\infty.
\] 
Later, we make these same assumptions on $\lambda_{j}$ and on the maximum integral $\max\limits_{1\le j\le M}\int |\log f_{j}^{0}(\vv u)|\,du$ to prove Theorem $1$.
\end{remark}
\begin{remark}
It is also true that any algorithm based on the updating mechanism proposed in Lemma ~\ref{basic_l} is going to be an MM algorithm. Indeed, for this to be true, there must exist a majorizing functional $b^{0}(\vv f,\bm\lambda)$ such that, when shifted by a constant, it majorizes the ``original" functional $l(\vv f,\bm\lambda)$:
\[
b^{0}(\vv f,\bm\lambda)+C^{0}\ge l(\vv f,\bm\lambda)
\] 
with equality when $(\vv f,\bm\lambda)=(\vv f^{0},\bm\lambda^{0})$ \cite{hunter2004tutorial, wu2010mm}. The exact form of the majorizing functional depends on the current parameter values and so does the constant $C^{0}$; this is why we are using the superscript for both $b^{0}(\vv f,\bm\lambda)$ and $C^{0}$. The difference
\begin{align*}
&l(\vv f,\bm\lambda)-l(\vv f^{0},\bm\lambda^{0})\le -\int g(\vv x)\sum_{j=1}^{m}w_{j}^{0}(\vv x)\log\left(\lambda_{j}{\cal N}_{j}f_{j}(\vv x)\right)\,d\vv x\\
&- \left(-\int g(\vv x)\sum_{j=1}^{m}w_{j}^{0}(\vv x)\log\left(\lambda_{j}^{0}{\cal N}_{j}f_{j}^{0}(\vv x)\right)\,d\vv x\right)=b^{0}(\vv f,\bm\lambda)-b^{0}(\vv f^{0},\bm\lambda^{0})
\end{align*}
where the functional $b^{0}(\vv f,\bm\lambda)$ is defined as 
\[
b^{0}(\vv f,\bm\lambda)=-\int g(\vv x)\sum_{j=1}^{m}w_{j}^{0}(\vv x)\log\left(\lambda_{j}{\cal N}_{j}f_{j}(\vv x)\right)\,d\vv x.
\]
With this definition in mind, we have immediately that  
\[
l(\vv f,\bm\lambda)\le b^{0}(\vv f,\bm\lambda)+C^{0}
\]
where $C^{0}=l(\vv f^{0},\bm\lambda^{0})-b^{0}(\vv f^{0},\bm\lambda^{0})$. Thus, we can confidently claim at this point that the algorithm that we proposed is, indeed, an MM algorithm
\end{remark}
The above result suggests the following algorithm that can be used to estimate parameters $\vv f$ and $\bm\lambda$ of the density mixture \eqref{model}. 
\begin{enumerate}
\item Initialize the algorithm with a choice of $\vv f^0,\bm\lambda^0$
\item \label{init}For $t=0,1,\ldots$ define weights
\begin{equation}\label{weight}
w_{j}^{(t)}(\vv x)=\frac{\lambda_{j}^{(t)}{\cal N}_{j}f_{j}^{(t)}(\vv x)}{{\cal M}_{\bm\lambda^{(t)}}{\cal N}\vv f^{(t)}(\vv x)}
\end{equation}
where $j=1,\ldots,m$
\item Define a new set of probabilities
\begin{equation}\label{lambda}
\lambda_{j}^{(t+1)}=\int g(\vv x)w_{j}^{(t)}(\vv x)\,d\vv x,
\end{equation}
$j=1,\ldots,m$
\item Define a new set of $d$-dimensional density functions
\begin{equation}\label{function}
f_{j}^{(t+1)}(\vv x)=\alpha_{j}^{(t+1)}\int K_{j,H}(\vv x-\vv u)\frac{g(\vv u){\cal N}_{j}f_j^{(t)}(\vv u)}{{\cal M}_{\bm\lambda^{(t+1)}}{\cal N}\vv f^{(t)}(\vv u)}\,d\vv u
\end{equation}
for $j=1,\ldots,m$ where $\alpha_j^{(t+1)}$ is a normalizing constant that depends on the step of iteration $t$. 
\item Go back to step \eqref{init}.
\end{enumerate}
\begin{remark}
Note that the objective function of this algorithm converges. Indeed, we already showed that it is monotonically decreasing with respect to the negative smoothed log-likelihood defined in \eqref{sm_lik}. It only remains to notice that this negative-smoothed log-likelihood function is bounded from below by zero. Indeed, due to the \eqref{KL_rep}, it can be represented as a sum of the KL distance between the target density $g(\vv x)$ and the convex combination of the smoothed component densities, on one hand, and the difference 
\[
1-\sum_{j=1}^{m}\lambda_j\int {\cal N}_{j}f_j(\vv x)\,d\vv x
\] on the other hand. The former is always non-negative. The latter is also non-negative since ${\cal N}_{j}f_{j}(\vv x)\le {\cal S}f_{j}(\vv x)$ due to Jensen's inequality
\end{remark}

\subsection{Extended algorithm for estimation of copula parameters}

In the previous section we obtained an MM algorithm capable of estimating both probability weights and the density functions of the general nonparametric density mixture model \eqref{model}. In this version, the proposed algorithm does not attempt to characterize the dependence mechanism between marginals of each density function $f_{j}(\vv x)$, $j=1,\ldots,m$. Let us now assume that every multivariate density function $f_{j}(\vv x)$ is represented as in \eqref{cmpt} with the assumption that the parametric form of each copula $C_{j}$ is known and only the value of the parameter $\rho_{j}$ is unknown. In the current subsection, we show that, in this case, the algorithm we have just derived can be modified to estimate the parameter $\rho_{j}$ in addition to the parameters discussed earlier. 

Let $\vv x\in \bR^{d}$ be an arbitrary vector with $x$ being the value of its $k$th coordinate. Then, the $k$th marginal density function of the $j$th component at the $t$th step of iteration, $f_{jk}^{(t)}(x)$, can be obtained from the (current) estimate of the joint density function $f^{(t)}_{j}(\vv x)$ as 
\begin{equation}\label{marg_density}
f_{jk}^{(t)}(x)=\int_{-\infty}^{\infty}\cdots \int_{-\infty}^{\infty}f_{j}^{(t)}(\vv x)dx_{1}\ldots,dx_{k-1}dx_{k+1}\ldots dx_{d}.
\end{equation}
Ordinarily, if the data are known, one can estimate marginal density functions directly from the data - no knowledge of the dependence structure between variables is needed. This is not the case here, however, since we don't have the advanced knowledge of which observations are assigned to which component of the model and, therefore, we cannot estimate marginal densities $f_{jk}(x)$ at first. As we proceed, obtaining more precise estimates of joint component density functions $f_{j}(\vv x)$ enables us to start estimating marginal density functions as well. 

Next, the knowledge of marginal density functions implies that we can also obtain the marginal cumulative density function $F_{jk}^{(t)}(x)=\int_{-\infty}^{x}f_{jk}^{(t)}(s)\,ds$. Let $\vv z=(z_{1},\ldots,z_{d})^{'}$ where $z_{k}=F_{jk}^{(t)}(x_k)$, $k=1,\ldots,d$. Then, using the representation \eqref{cmpt}, we can estimate the $j$th copula density function at $\vv z$ at the $t$th step of iteration as
\begin{equation}\label{copula_density}
c_{j}^{(t)}(\vv z):=\frac{f_j^{(t)}(\vv x)}{\prod_{k=1}^{d}f_{jk}^{(t)}(x_k)}.
\end{equation}

%Now, let $\vv y=(y_1,\ldots,y_d)^{'}\in [0,1]^{d}$ be a point at which we want to estimate the copula density. Note that, if we know the joint density at the step of iteration $t$ and all of its marginal densities then, according to \eqref{copula_density}, we can estimate the copula density at the point with coordinates $\left(F_{j1}^{(t)}(x_1),\ldots,F_{jd}^{(t)}(x_d)\right)$ only.
%To be able to estimate the copula density at an arbitrary point $\vv y\in [0,1]^{d}$, we remember that, first, we are working with continuous marginal densities only. Hence, we  can, for any point $\vv y=(y_1,\ldots,y_d)^{'}\in [0,1]^{d}$, define $\vv z =(z_{1},\ldots,z_{d})^{'}$ %where
%\begin{equation}\label{eq:trans}
%z_{k}=(F_{jk}^{(t)})^{-1}(y_k)
%\end{equation}
% for $k=1,\ldots,d$ and $(F_{jk}^{(t)})^{-1}(y_k)$ is the inverse CDF (quantile function) of the %distribution of the random variable $X_{jk}$.
Finally, let $c(\vv z; \rho_{j})$ be the true $j$th copula density function, viewed as a function of $\rho_j$ at the point $\vv z$. If there are $n$ observations $\vv x_1,\ldots,\vv x_n$ from the model \eqref{model}, we can define  
\[
\rho_{j}^{(t)}=\arg\min S_{n}^{(t)}(\rho_{j})
\]
where 
\[
S_{n}^{(t)}(\rho_{j}) =\frac{1}{n}\sum_{i=1}^{n}[c(\vv z_i;\rho_{j})-c_{j}^{(t)}(\vv z_i)]^{2}.
\]

This discussion suggests the following extended algorithm. We denote $\bm\rho=(\rho_{1},\ldots,\rho_{m})^{'}$.  
\begin{enumerate}
\item Initialize the algorithm with a choice of $\vv f^0,\bm\lambda^0,\bm\rho^{0}$
%\item Find individual marginal densities $f_{jk}^{0}(x)$ for all of $k=1,\ldots,d$ using marginal integration as in \eqref{marg_density}
%\item Obtain individual marginal cumulative densities $F_{jk}^{0}(x)=\int_{-\infty}^{x}f_{jk}^{0}(s)\,ds$ 
%\item Let $(F_{jk}^{(t+1)})^{-1}(y):=\inf\{s: F_{jk}^{(t+1)}(s)\ge y\}$ be the inverse cumulative density function of $F_{jk}^{(t+1)}$. 
%Define points $\vv x_i=(x_{i1},\ldots,x_{id})^{'}$ where $x_{ik}=(F_{jk}^{(t)})^{-1}(y_{ik})$ for $i=1,\ldots,n$ and $k=1,\ldots,d$.
\item \label{init_upd}For $t=0,1,\ldots$ define weights
\begin{equation}\label{weight_upd}
w_{j}^{(t)}(\vv x)=\frac{\lambda_{j}^{(t)}{\cal N}f_{j}^{(t)}(\vv x)}{{\cal M}_{\bm\lambda^{(t)}}{\cal N}\vv f^{(t)}(\vv x)}
\end{equation}
where $j=1,\ldots,m$
\item Find individual marginal densities $f_{jk}^{(t)}(x_{ik})$ for all of $k=1,\ldots,d$ using marginal integration as in \eqref{marg_density}
\item Obtain individual marginal cumulative densities $F_{jk}^{(t)}(x_{ik})=\int_{-\infty}^{x_{ik}}f_{jk}^{(t)}(s)\,ds$ 
%\item Let $(F_{jk}^{(t)})^{-1}(y):=\inf\{s: F_{jk}^{(t)}(s)\ge y\}$ be the inverse cumulative density function of $F_{jk}^{(t)}$. 
%Define points $\vv z_i=(z_{i1},\ldots,z_{id})^{'}$ where $z_{ik}=(F_{jk}^{(t)})^{-1}(y_{ik})$ %for $i=1,\ldots,n$ and $k=1,\ldots,d$.
\item Define a new set of probabilities
\begin{equation}\label{lambda_upd}
\lambda_{j}^{(t+1)}=\int g(\vv x)w_{j}^{(t)}(\vv x)\,d\vv x,
\end{equation}
$j=1,\ldots,m$
\item Define a new set of $d$-dimensional density functions at the set of points 
\begin{equation}\label{function_upd}
f_{j}^{(t+1)}(\vv x)=\alpha_{j}^{(t+1)}\int K_{j,H}(\vv x_i-\vv u)\frac{g(\vv u){\cal N}f_j^{(t)}(\vv u)}{{\cal M}_{\bm\lambda^{(t+1)}}{\cal N}\vv f^{(t)}(\vv u)}\,d\vv u
\end{equation}
for $j=1,\ldots,m$ where $\alpha_j^{(t+1)}$ is a normalizing constant that depends on the step of iteration $t$. 
%\item Find individual marginal densities $f_{jk}^{(t+1)}(x_{ik})$ for all of $k=1,\ldots,d$ using marginal integration as in \eqref{marg_density}
%\item Obtain individual marginal cumulative densities $F_{jk}^{(t+1)}(x_{ik})=\int_{-\infty}^{x_{ik}}f_{jk}^{(t+1)}(s)\,ds$ 
%\item Let $(F_{jk}^{(t+1)})^{-1}(y):=\inf\{s: F_{jk}^{(t+1)}(s)\ge y\}$ be the inverse cumulative density function of $F_{jk}^{(t+1)}$. 
%Define points $\vv x_i=(x_{i1},\ldots,x_{id})^{'}$ where $x_{ik}=(F_{jk}^{(t)})^{-1}(y_{ik})$ for $i=1,\ldots,n$ and $k=1,\ldots,d$.
\item  Let $\vv z_{i}=(z_{i1},\ldots,z_{id})^{'}$ where $z_{ik}=F_{jk}^{(t)}(x_{ik})$. Estimate the copula density of the $j$th component at the point $\vv z_i$ at the step of iteration $t$ as  
\[
c_{j}^{(t+1)}\left(\vv z_{i}\right):=\frac{f_j^{(t+1)}(\vv x_i)}{\prod_{k=1}^{d}f_{jk}^{(t+1)}(x_{ik})}. 
\] 
\item Let $c(\vv z_{i};\rho_{j})$ be the value of the true $j$th copula density function at the point $\vv z_i$. We now estimate the value of the copula parameter $\rho_j$ as the solution of an optimization problem
\[
\rho^{(t+1)}_{j}=\arg\min S_{n}^{(t+1)}(\rho_{j})
\]
where $S_{n}^{(t+1)}(\rho_{j}) =\frac{1}{n}\sum_{i=1}^{n}[c(\vv z_i;\rho_j)-c_{j}^{(t+1)}(\vv z_i)]^{2}$.
\item Go back to step \eqref{init_upd}.
\end{enumerate}
As opposed to the general algorithm setting where the kernel function could be an essentially arbitrary multivariate rescaled density function, here we use the value of the $j$th component density function from the first iteration $f_{j}^{(0)}$ as a kernel. This kernel undergoes what amounts to a reweighting procedure at every step of iteration in \eqref{function_upd} to become the next iteration $f_{j}^{(1)}$, $f_{j}^{(2)}$ etc. In other words, we define
\begin{equation}\label{cp-kernel}
K_{j}(\vv x)=c(F_{j1}^{(0)}(x_1),\ldots,F_{jd}^{(0)}(x_d);\rho_{j}^{0})\prod_{k=1}^{d}f_{jk}^{(0)}(x_k).
\end{equation}
We are choosing this specific kernel in order to guarantee that the next iteration preserves the specific parametric copula structure defined by the known copula density function $c(\cdot,\ldots,\cdot)$ that is assumed to be known in advance.  Note that, although available at the step of iteration $t$, the updated copula parameter value $\rho_{j}^{(t)}$ and the updated values of marginal density functions are not used in the definition of the kernel function and the initialized values are used instead. The reason is the same reason that makes the algorithm  of \cite{Biometrika_2011LevineHunterChauveau} and algorithms presented in \cite{chauveau2015semi} use a fixed bandwidth at each step of iteration - preservation of monotonicity. In principle, it does seem reasonable to update the bandwidth matrix $H$, the copula parameter $\rho_{j}$, all of the estimated marginal density functions and the estimated cumulative distribution functions in the proposed kernel function at each step of iteration since this amounts to using the new information available about estimated component densities. This, however, makes the objective functional \eqref{sm_lik} ill-defined and, therefore, invalidates the desirable monotonicity.

\subsection{Practical implementation}\label{pr_alg}
The extended algorithm proposed above is a ``conceptual" algorithm that assumes the knowledge of the target density function $g(\vv x)$. Of course, this is not the case in reality and so we propose a discretized version of this algorithm. Since we have a sample of observations $\vv x_1,\ldots,\vv x_n$ that is assumed to have been generated from $g(\vv x)$,  we substitute the integral $\int g(\vv x)w_{j}^{(t)}(\vv x)\,d\vv x$ with $\frac{1}{n}\sum_{i=1}^{n}w_{j}^{(t)}(\vv x_i)$. To make our exposition easier to follow, we will give a detailed form of our algorithm in its discretized version. To prevent the label switching, we are working under the constraint $\lambda_{1}\le \lambda_{2}\le \cdots\le \lambda_{m}$.

\begin{enumerate}
\item Initialize the algorithm with a choice of $\vv f^0,\bm\lambda^0$, $\bm\rho^{0}$
\item For $t=0,1,\ldots$ find individual marginal densities $f_{jk}^{(t)}(x)$ for all of $k=1,\ldots,d$ using a discretized integral from \eqref{marg_density}
\item \label{init_upd_disc} Obtain individual marginal cumulative density functions $F_{jk}^{(t)}(x)$ by discretizing $\int_{-\infty}^{x}f_{jk}^{(t)}(s)\,ds$ 
%\item  Let $(F_{jk}^{(t)})^{-1}(y):=\inf\{s: F_{jk}^{(t)}(s)\ge y\}$ be the inverse cumulative density function of $F_{jk}^{(t)}$. For any $\vv y_i$, define $\vv z_{i} =(z_{i1},\ldots,z_{id})^{'}$ where $z_{ik}=(F_{jk}^{(t)})^{-1}(z_{ik})$ for $k=1,\ldots,d$
\item For $t=0,1,\ldots$ define weights
\begin{equation}\label{weight_upd_disc}
w_{j}^{(t)}(\vv x_i)=\frac{\lambda_{j}^{(t)}{\cal N}_{j}f_{j}^{(t)}(\vv x)}{{\cal M}_{\bm\lambda^{(t)}}{\cal N}\vv f^{(t)}(\vv x)}
\end{equation}
where $j=1,\ldots,m$
\item Define a new set of probabilities
\begin{equation}\label{lambda_upd_disc}
\lambda_{j}^{(t+1)}=\frac{1}{n}\sum_{i=1}^{n}w_{j}^{(t)}(\vv x_i)
\end{equation}
$j=1,\ldots,m$
\item Define a new set of $d$-dimensional density functions estimated at $\vv x_l$, $l=1,\ldots,n$ 
\begin{equation}\label{function_upd_disc}
f_{j}^{(t+1)}(\vv x_l)=\frac{\sum_{i=1}^{n}K_{j,H}(\vv x_{l}-\vv x_i)\frac{{\cal N}_{j}f_j^{(t)}(\vv x_i)}{{\cal M}_{\bm\lambda^{(t+1)}}{\cal N}\vv f^{(t)}(\vv x_i)}}{\sum_{i=1}^{n}\frac{{\cal N}_{j}f_j^{(t)}(\vv x_i)}{{\cal M}_{\bm\lambda^{(t+1)}}{\cal N}\vv f^{(t)}(\vv x_i)}}
\end{equation}
for $j=1,\ldots,m$ 
where the kernel $K_{j}(\cdot)$ is selected as in \eqref{cp-kernel}
\item \item Obtain individual marginal densities $f_{jk}^{(t+1)}(x)$ for all of $k=1,\ldots,d$ using a discretized integral from \eqref{marg_density}
\item\label{est_copula}  Let $\vv z_{i}=(z_{i1},\ldots,z_{id})^{'}$ where $z_{ik}=F_{jk}^{(t)}(x_{ik})$, $k=1,\ldots,d$. Estimate the copula density of the $j$th component at the point $\vv z_i$ at the step of iteration $t$ as  
\[
c^{(t+1)}(\vv z_{l}):=\frac{f_j^{(t+1)}(\vv z_l)}{\prod_{k=1}^{d}f_{jk}^{(t+1)}(x_{lk})}. 
\] 
\item\label{copula_opt} Let $c(\vv y_{l};\rho_{j})$ be the value of the true $j$th copula density function at the point $\vv y_l$ if $\rho^{(j)}$ is the true value of the copula parameter. We now estimate the value of the copula parameter $\rho_j$ at the $t+1$st step of iteration as the solution of an optimization problem
\[
\rho^{(t+1)}_{j}=\arg\min S_{n}^{(t+1)}(\rho_{j})
\]
where $S_{n}^{(t+1)}(\rho_{j}) =\frac{1}{n}\sum_{l=1}^{n}[c(\vv y_l;\rho_j)-c^{(t+1)}(\vv y_l)]^{2}$.
\item Go back to step \eqref{init_upd_disc}.
\end{enumerate}

At every step of iteration, it is also necessary to choose a numerical integration approach to obtain updated marginal densities $f_{jk}^{(t)}(x)$. Also, the algorithm requires performing a multivariate convolution every time an operator ${\cal N}$ is applied. Note that there are few flexible copula families available for higher dimensions. Some possibilities include using e.g. vine copulas \cite{aas_czado} or factor copulas \cite{kr_joe}. However, attempting to use them results in a very difficult optimization problem at every iteration step whenever the new value of the copula parameter has to be obtained. Due to this, we only consider the low-dimensional case in this manuscript. In a low-dimensional case, one reasonable approach to numerical integration is to use e.g. adaptive Gaussian quadrature to integrate out all the unneeded coordinates in $f_{j}^{(t)}(\vv x)$. In particular, we use a modern R package {\it cubature} \cite{narasimhan2025package} that uses an adaptive multidimensional Gaussian quadrature and works with vector-valued integrands over hypercubes. This approach is used both to obtain estimated marginal densities $f_{jk}^{(t)}(x)$ and to compute numerical convolutions. Also, to simplify the exposition, we only consider copulas parameterized by a one-dimensional parameter. Using such copulas is typically enough whenever the dependence structure of the data can be adequately represented with a single measure of association such as, for example, Kendall's tau or Spearman's rho (see e.g. \cite{nelsen2007introduction} Section $4$ dedicated to Archimedean copulas). In the univariate  copula parameter $\rho_{j}$ case, the minimization of the criterion $S_{n}^{(t+1)}$ can be performed using a basic method which is a combination of the golden section search and successive parabolic interpolation that is encoded in the R function {\it optimize}. The algorithm stops when the relative difference of successive absolute values of the discretized objective functional $-\sum_{i=1}^{n}({\cal M}_{\bm\lambda^{(t)}}{\cal N}f^{(t)})(x_i)$ become smaller than the prespecified tolerance value of $\eps=10^{-8}$.

\section{Convergence of functional sequences generated by the algorithm}\label{Con_func}

In addition to the convergence of the objective function, we also investigate if a sequence $\{\lambda_{j}^{(t)},f_{j}^{(t)}(\vv x)\}_{j=1}^{m}$, $t=0,1,2,3,\ldots$, generated by the algorithm, converges as well. Such convergence results have been established earlier for related algorithms in e.g. \cite{Biometrika_2011LevineHunterChauveau} and \cite{shen2018mm}. For conciseness, we will also use the notation $\left\{\bm\lambda^{(t_{l})},\vv f^{(t_{l})}\right\}$ for this sequence. The following result can be established. 
\begin{theorem}\label{Th_con}
Let there exist $L>0$ such that $|K_{j,H}(\vv x)-K_{j,H}(\vv y)|\le L||\vv x-\vv y||$ for any $\vv x,\vv y\in \Omega$ and $j=1,\ldots,m$; in other words, we assume that the kernel density functions $K_{j}(\cdot)$ are Lipschitz continuous. Moreover, we also assume that all of them are  bounded away from zero. Finally, we assume that all of the initial $\lambda_j^0>0$ and denote $\lambda=\min_{1\le j \le m}\lambda_j^0$. Then, there exists a subsequence $\left\{\bm\lambda^{(t_{l})},\vv f^{(t_{l})}\right\}$, $l=1,2,\ldots$ that converges to a local minimum of the functional \eqref{sm_lik}.
%We also assume that all of the $m$ densities of the algorithm are log integrable. With this mind, let $B=\max_{1\le j \le m}\int|\log f_j^0(\vv u)|\,d\vv u<\infty$. 
\end{theorem}
\begin{remark}
Since kernel function $K_{j}(\cdot)$ are products of marginal density functions at initialization and the copula density function, boundedness away from zero requires that all of the initial marginal density functions and a copula density function must be bounded away from zero. An example of non-trivial family of copula density functions bounded away from zero can be found in e.g. \cite{longla2024new} p. $4338-39$. The family introduced there is a perturbation of an independence copula which is a usual (trivial) example of a copula whose density is bounded away from zero. Note also that thus introduced family of bivariate copulas $C_{1,\lambda}(u,v)$ (in the notation of \cite{longla2024new}) is quite suitable for practical use since it is capable of modeling a pair of random variables with the full range of Pearson correlation from $-1$ to $1$.  The same copula family is also Lipschitz continuous on a compact set. Note that one can also choose initialized marginal density functions to be bounded away from zero. Hence, as long as the initialized marginal density functions are Lipschitz continuous on $\Omega$, the entire proposed kernel is Lipschitz continuous as well (because the product of Lipschitz continuous functions on a compact is Lipschitz continuous on that compact). 
\end{remark}
\begin{proof}
The proof of this Theorem consists of the combination of proofs of Lemmas ~\ref{sub} and  ~\ref{lw_semi} below.
\end{proof}.
\begin{lemma}\label{sub}
Under the assumptions of Theorem \eqref{Th_con} there exists a converging subsequence of the sequence $f_j^{(t)}$ for any  $j=1,\ldots,m$, $t=0,1,2,\ldots$ generated by the proposed algorithm. 
\end{lemma}
\begin{proof}
This result can be proved using Arzel\`{a}-Ascoli theorem. To do so, we have to show, first, that the sequence $\{f_j^{(t)}(\vv u)\}$ is uniformly bounded for all $t=0,1,\ldots$: there exists $F_1>0$ such that $|f_j^{(t)}(\vv u)|\le F_1$. Second, we have to show that the sequence $\{f_j^{(t)}\}$ is also uniformly continuous: there exists $F_2>0$ such that $|f_j^{(t)}(\vv x)-f_j^{(t)}(\vv y)|\le F_2||\vv x-\vv y||$ for any $\vv x,\vv y\in \Omega$ and $t=0,1,2,\ldots$. 
 
We begin with establishing a simple fact that ${\cal M}_{\lambda^{(t)}}{\cal N}f^{(t-1)}(\vv u)$ is bounded from above for any step of iteration $t$. This is the quantity that is found in the denominator of the integral that defines the next iteration of the function $f_j$, $j=1,\ldots,m$ according to \eqref{function}. First, remember that all of the kernel functions are uniformly continuous on $\Omega$. Also, all of them are also strictly positive on a compact set. Therefore, there exist positive constants $a$ and $A$ such that $a\le K_{j,H}(\cdot)\le A$. 
Indeed, using Jensen's inequality, one can easily verify that ${\cal N}_{j}f_j^{(t-1)}(\vv u)\le \int K_{j,H}(\vv x-\vv u)f_j^{(t-1)}(\vv x)\,d\vv x={\cal S}_{j}f^{(t-1)}_j(\vv u)$. Therefore, 
\begin{align}\label{up_bound}
&{\cal M}_{\lambda^{(t)}}{\cal N}f^{(t-1)}(\vv u)=\sum_{j=1}^{m}\lambda_j^{(t)}{\cal N}_{j}f_j^{(t-1)}(\vv u)\le \int K_{j,H}(\vv x-\vv u)\sum_{j=1}^{m}\lambda_j^{(t)} f_j^{(t-1)}(\vv x)\,d\vv x\\
&\le A \int \sum_{j=1}^{m}\lambda_j^{(t)} f_j^{(t-1)}(\vv x)\,d\vv x=A\nonumber
\end{align}
because $\sum_{j=1}^{m}\lambda_j^{(t)} f_j^{(t-1)}(\vv x)$ is a density and so its integral is equal to $1$. Note also that the same argument produces boundedness of ${\cal M}_{\lambda^{(t)}}{\cal N}f^{(t)}(\vv u)$, which is the quantity in the denominator of the weight $w_j^{(t)}$ defined in \eqref{weight} for any $t$ as well.

Next,we show that ${\cal N}_{j}f_j^{(t)}(\vv u)$ is also bounded away from zero uniformly for every step of iteration $t$ by mathematical induction. Due to boundedness away from zero of all initialized densities $f_{j}^{(0)}$, $j=1,\ldots,m$,  all of these densities are also log integrable. With this mind, let $B=\max_{1\le j \le m}\int|\log f_j^0(\vv u)|\,d\vv u<\infty$. First, for $t=0$ we have $\int\left\vert\log f_j^0(\vv u)\right\vert\,d\vv u\le B$ for any $j=1,\ldots,m$ by assumption. Due to this, we clearly have ${\cal N}_{j}f_j^0(\vv u)\ge e^{-AB}>0$. Next, we assume that ${\cal N}_{j}f_j^{(t-1)}(\vv u)\ge a_1>0$ for some $a_1$. Then, we conclude that, since ${\cal M}_{\lambda^{(t)}}{\cal N}f^{(t-1)}(\vv u)\le A$, we also have  immediately that, due to \eqref{function}, $f^{(t)}_j(\vv u)$ is also bounded away from zero uniformly for any $t$. But that, of course, immediately implies that ${\cal N}_{j}f_j^{(t)}(\vv u)$ is bounded away from zero. Thus, by induction, ${\cal N}_{j}f_j^{(t)}(\vv u)$ is bounded away from zero for an arbitrary step of iteration $t$. 

Next, we are going to prove that $\lambda_{j}^{(t)}$ is  bounded away from zero at any step of iteration $t$ using the mathematical induction. First, since $\lambda_j^{0}\ge \lambda$, we find immediately that $w_j^0\ge \frac{\lambda e^{-AB}}{A}>0$ for any $j=1,\ldots,m$. To use the induction principle, again assume that $\min_{1\le j \le m}\lambda_j^{(t-1)}\ge \tilde\lambda>0$. Then, since $\lambda_{j}^{(t-1)}$ is bounded away from zero by assumption, and ${\cal N}_{j}f_{j}^{(t-1)}(u)$ is also assumed to be bounded away from zero, we immediately conclude that the weight $w_j^{(t-1)}$ is bounded away from zero uniformly. Due to \eqref{lambda}, weights $\lambda_j^{(t)}$ are also uniformly bounded away from zero and the induction is complete. 

The uniform boundedness of all functions $f_j^{(t+1)}(\vv u)$ should now be clear. First, note that at any step of iteration probability weights $\lambda^{(t+1)}_{j}$ are bounded away from zero and the same is true about ${\cal N}_{j}f_{j}(\vv u)$. Thus, the denominator ${\cal M}_{\lambda^{(t+1)}}{\cal N}f^{(t)}(\vv u)$ is bounded away from zero by some positive $M$. Therefore, 
\begin{align*}
&f_j^{(t+1)}(\vv u)\le \frac{1}{M}\alpha_j^{(t+1)}\int K_{j,H}(\vv x-\vv u)g(\vv u){\cal N}_{j}f_{j}^{(t)}(\vv u)\,d\vv u\\
&\le \frac{1}{M}\alpha_{j}^{(t+1)}\int K_{j,H}(\vv x-\vv u)g(\vv u)\left\{\int K_{j,H}(\vv z-\vv u)f_{j}^{(t)}(\vv z)\,d\vv z\right\}\,d\vv u
\end{align*}
The double integral above is uniformly bounded from above because both $g(\vv u)$ and $f_j^{(t)}(\vv z)$ are densities while $K_{j,H}(\cdot)\le A$. Thus, we proved that there exists a constant $F$ such that $\left\vert f_{j}^{(t+1)}(\vv u)\right\vert \le F$. 

To show the uniform equicontinuity, we consider the difference
\begin{align*}
&|f_{j}^{(t)}(\vv x)-f^{(t)}_j(\vv y)|\le \alpha_j^{(t)}\int \frac{g(\vv u){\cal N}_{j}f_j^{(t-1)}(\vv u)}{{\cal M}_{\bm\lambda^{(t)}}{\cal N}\vv f^{(t-1)}(\vv u)}\left\vert K_{j,H}(\vv x-\vv u)-K_{j,H}(\vv y-\vv u)\right\vert\,d\vv u\\
&\le L\alpha_j^{(t)}|\vv x-\vv y|\int\frac{g(\vv u){\cal N}_{j}f_j^{(t-1)}(\vv u)}{{\cal M}_{\bm\lambda^{(t)}}{\cal N}\vv f^{(t-1)}(\vv u)}\,d\vv u.
\end{align*}
The above expression is clearly bounded from above uniformly due to the boundedness of ${\cal M}_{\bm\lambda^{(t)}}{\cal N}\vv f^{(t-1)}(\vv u)$ away from zero, and the fact that $K_{j,H}(\cdot)\le A$ on $\Omega$. Thus, there exists a constant $F_2$ such that
\[
|f_{j}^{(t)}(\vv x)-f^{(t)}_j(\vv y)|\le F_2|\vv x-\vv y|
\]
where $F_2$ does not depend on $t$. The result has been proved. 
\end{proof}
As a next step, note that every function in the sequence $f_{j}^{(t)}$ except, perhaps, the first one, can be represented as 
\[
f_{j}^{(t+1)}(\vv x)={\cal S}\phi_{j}(\vv x)
\]
where $0\le \phi_{j}(\vv x)\in L_{1}(\Omega)$ and $\int_{\Omega}\phi_{j}(\vv x)\,d\vv x=1$. In this representation, 
\[
\phi_{j}(\vv x)=\alpha_{j}^{(t+1)}\frac{g(\vv x){\cal N}_{j}f_{j}^{(t)}(\vv x)}{{\cal M}_{\bm\lambda^{(t+1)}}{\cal N}\vv f^{(t)}(\vv x)}
\]
and $\int_{\Omega}\phi_{j}(\vv x)\,d\vv x=1$ due to the definition of $\alpha_{j}^{(t+1)}$. Let us denote 
\[
B=\{{\cal S}\phi:0\le \phi_{j}\in {\cal F},\int \phi_{j}(\vv x)\,d\vv x=1\}.
\]
Then, using the same argument as in \cite{Biometrika_2011LevineHunterChauveau}, we can show immediately that ${\cal N}\vv f(\vv x)$ is well defined for any $\vv f\in B$ and, therefore, the entire functional $l(\theta)$ is well defined on $B$ as well. Now we can establish the following Lemma.
\begin{lemma}\label{lw_semi}
Under the assumptions of Theorem \eqref{Th_con} the functional \eqref{sm_lik} is lower semicontinuous on $B$. 
\end{lemma}
\begin{proof}
First, note that by definition, $\lambda_{j}^{(t)}\le 1$ for any step of iteration $t$, $j=1,\ldots,m$. In the proof of Lemma \eqref{sub}, we showed that $\lambda_{j}^{(t)}$ is also uniformly bounded away from zero. By Bolzano-Weierstrass theorem, $\{\lambda_{j}^{(t)}\}$ must, then, have a convergent subsequence - e.g. a subsequence $\lambda_{j}^{(t_{l})}$ that converges to the (finite) $\liminf\lambda_{j}^{(t)}$. Let us denote this lower limit $\xi_{j}$ for the sequence $\lambda_{j}^{(t)}$ and $\bm\xi=(\xi_1,\ldots,\xi_{m})^{'}$. Now, consider a subsequence $f_{j}^{(t_{l})}$, $l=1,2,\ldots$ that converges to $\liminf f_{j}^{(t)}$. Note that the result of Lemma \eqref{sub} can be applied not only to the entire sequence $f_{j}^{(t)}$ but also to the subsequence $f_{j}^{(t_{l})}$; therefore, one can conclude that this subsequence has a uniformly converging subsubsequence that converges uniformly to the same limit as the subsequence $f_{j}^{(t_{l})}$ as $l\rightarrow \infty$. Thus, we can assume, without loss of generality, that the entire subsequence $f_{j}^{(t_{l})}\rightarrow \liminf f_{j}^{(t)}\equiv \psi_j$ uniformly as $l\rightarrow \infty$.  Since every element of the subsequence is bounded away from zero, we have $\log f_{j}^{(t_{l})}\rightarrow \log \psi_{j}$. Considering a joint sequence $(\lambda_{j}^{(t_{l})}, f_{j}^{(t_{l})})$ for all $j=1,\ldots,m$ and following the argument of \cite{Biometrika_2011LevineHunterChauveau} we quickly conclude that ${\cal M}_{\vv\lambda^{(t_{l}+1)}}{\cal N}\vv f^{(t_{l})}\rightarrow {\cal M}_{\xi}{\cal N}\bm\psi$. By an elementary inequality, $\rho(t)=t-\log t-1\ge 0$. Moreover, since ${\cal M}_{\xi}{\cal N}\bm\psi(\vv x)\ge 0$, by Fatou lemma it is always true that 
\[
\int g(\vv x){\cal M}_{\bm\xi}{\cal N}\bm\psi(\vv x)\le \liminf \int g(\vv x){\cal M}_{\bm\lambda^{(t_{l}+1)}}{\cal N}\vv f^{(t_{l})}(\vv x).
\]
Due to the above, the inequality 
\[
\int g(\vv x)\rho({\cal M}_{\bm\xi}{\cal N}\bm\psi(\vv x))\,d\vv x\le \liminf \int g(\vv x)\rho({\cal M}_{\bm\lambda^{(t_{l}+1)}}{\cal N}\vv f^{(t_{l})}(\vv x))\,d\vv x.
\]
is equivalent to 
\[
-\int g(\vv x)\log \{{\cal M}_{\bm\xi}{\cal N}\psi(\vv x)\}\,d\vv x \le \liminf \int g(\vv x)\log \{
{\cal M}_{\bm\lambda^{(t_{l}+1)}}{\cal N}\vv f^{(t_{l})}(\vv x)\}\,d\vv x
\]
which implies that the functional \eqref{sm_lik} is indeed a lower semicontinuous one.  
\end{proof}

One can also prove the result of Theorem \eqref{Th_con} in the case where the unknown density domain $\Omega$ is not compact as long as the smoothing kernel is defined on a compact. To show that this is the case, we begin with establishing an alternative version of Lemma $1$.  In order to prove this alternative version, we must use Fr\'{e}chet-Kolmogorov theorem instead of the Arzel\`{a}-Ascoli theorem. The Fr\'{e}chet-Kolmogorov theorem establishes necessary and sufficient conditions of relative compactness for a functional space $L^{p}(\bR^{n})$; in our case, $p=1$. Since it is a little less well known than Arzel\`{a}-Ascoli theorem, and also for convenience, we state this theorem in full. 
\begin{theorem}\label{Fr-Kolm}
Let ${\cal B}$ be a set in $L^p(\bR^{n})$ with $p\in [1,\infty)$. The subset ${\cal B}$ is relatively compact if and only if the following properties hold for any function $f\in {\cal B}$: 
\begin{enumerate}
\item $\lim_{r\rightarrow \infty}\int_{||x||>r}|f|^p=0$ uniformly on ${\cal B}$,
\item $\lim_{a\rightarrow 0}||\tau_af-f||_p=0$ uniformly on ${\cal B}$ 
\end{enumerate}
where $||\cdot||_{p}$ stands for the $L^{p}(\bR^{n})$ norm and where $\tau_af$ denotes the translation of $f$ by $a$, that is, $\tau_af(x)=f(x-a)$.
\end{theorem} 
A very nice proof of this result can be found in e.g. an expository paper \cite{hanche2010kolmogorov}. Now we can start with establishing the new version of Lemma \eqref{sub}.

\begin{lemma}\label{sub_noncom}
Let the assumptions of Theorem \eqref{Th_con} be true. We also assume that the kernel $K$ has a closed rectangular support $\Delta$ on which it is bounded away from zero. Then, there exists a converging subsequence of the sequence $f_j^{(t)}$ for any  $j=1,\ldots,m$, $t=0,1,2,\ldots$ generated by the proposed algorithm. 
\end{lemma}
\begin{remark}
Note that here we also have the freedom to define initialized density functions $f_{j}^{(0)}$ on a (possibly large) $d$-dimensional compact set and to make them strictly positive and Lipschitz continuous. In that case, the kernel function \eqref{cp-kernel} is going to have a compact support and be strictly positive  as long as the copula density function is a strictly positive one such as e.g. the one corresponding to $C_{1,\lambda}$ of \cite{longla2024new}. Hence, the assumptions of Lemma \eqref{sub_noncom} are going to be valid.
\end{remark}
\begin{proof}

Let us now check the first condition in the Theorem \eqref{Fr-Kolm} for functions that belong to the set $B$. As a preliminary step, we  show first that at any step of iteration $\lambda_{j}^{(t)}$ is bounded away from zero. Indeed, from our algorithm we can see that $\lambda_{j}^{(t)}=\int g(\vv x)w_{j}^{(t-1)}(\vv x)\,d\vv x$; thus, if we show that the weight $w_{j}^{(t-1)}(\vv x)$ is always bounded away from zero for any $t$, the probability $\lambda_{j}^{(t)}$ is bounded away from zero as well. Since the kernel function $K$ is bounded from below, we can easily claim that for any $f_{j}^{(t)}\in B$ $f_{j}^{(t)}(\vv x)=\int K_{j,H}(\vv x-\vv u)\phi_{j}(\vv u)\,d\vv u\ge \inf_{\vv x\in \Omega}K_{j,H}(\vv x-\vv u)\int \phi(\vv u)\,d\vv u=K_{*}>0$. Next, by definition of the smoothing operator ${\cal N}_{j}f_{j}^{(t)}(\vv x)$, and since $f_{j}^{(t)}\in B$ for any step of iteration $t$,  we have ${\cal N}_{j}f_{j}^{(t)}(\vv x)=\exp\left\{\int K_{j,H}(\vv x-\vv u)\log f_{j}^{(t)}(\vv u)\,d\vv u\right\}\ge \exp\{\log K^{*}\int K_{j,H}(\vv x-\vv u)\,d\vv u\}=K^{*}>0$ since the kernel function is a proper density function. Now, recall that at each step $t$ the weight is 
\begin{equation}
w_{j}^{(t-1)}(\vv x)=\frac{\lambda_{j}^{(t-1)}{\cal N}_{j}f_{j}^{(t-1)}(\vv x)}{{\cal M}_{\bm\lambda^{(t-1)}}{\cal N}\vv f^{(t-1)}(\vv x)}.
\end{equation}
Since the kernel $K$ is continuous on a compact rectangle, it is also necessarily bounded from above by a positive constant $M$. Using Jensen's inequality, we immediately find that ${\cal N}_{j}f_{j}^{(t-1)}(\vv x)$ is always bounded from above by $M$ as well. 
%Hence, the denominator of the integrand in the definition of $w_{j}^{(t-1)}(\vv x)$ is bounded from above and, therefore, $w_{j}^{(t-1)}(\vv x)$ is always bounded from below as long as $\lambda_{j}^{t-1}$ is bounded from below. Using the above argument, it is easy to see that as long as we start with $\lambda_{j}^{0}>0$ for any $j$, $\lambda_{j}^{t-1}$ will stay bounded away from zero at every step of iteration: $\lambda_{j}^{(t-1)}\ge l_{j}^{(t-1)}>0$ for some positive $l_{j}^{(t-1)}$  

Analyzing the constant $\alpha_{j}^{(t)}$, we see that
\begin{align*}
&\alpha_{j}^{(t)}=\frac{1}{\int d\vv x\int K_{j,H}(\vv x-\vv u)\frac{g(\vv u){\cal N}_{j}f_{j}^{(t)}(\vv u)}{{\cal M}_{\bm\lambda^{(t+1)}}{\cal N}\vv f^{(t)}(\vv u)}\,d\vv u}\le \frac{M}{K^{*}}\frac{1}{\int d\vv x\int K_{j,H}(\vv x-\vv u)g(\vv u)\,d\vv u}.
\end{align*}
The convolution $\int K_{j,H}(\vv x-\vv u)g(\vv u)\,d\vv u$ is always integrable as long as $K_{j,H}(\cdot)$ and the density function $g(\vv u)$ are integrable themselves since 
\[
\int d\vv x\int K_{j,H}(\vv x-\vv u)g(\vv u)\,d\vv u\le \int K_{j,H}(\vv x)\,d\vv x\int g(\vv u)\,d\vv u=1
\]
due to Young's convolution inequality. Thus, we can claim that $\alpha_{j}^{(t)}$ is bounded by a finite constant that does not depend on $t$. Let us denote this constant ${\cal C}$.

This implies that 
\begin{align*}
&\int_{||\vv x||>r}|f_{j}^{(t)}(\vv x)\,dx\le {\cal C}\int_{||\vv x||>r}\int K_{j,H}(\vv x-\vv u)\frac{g(\vv u){\cal N}_{j}f_{j}^{(t)}(\vv u)}{{\cal M}_{\bm\lambda^{(t+1)}}{\cal N}\vv f^{(t)}(\vv u)}\,d\vv u\\
&\le {\cal C}\frac{M}{K^{*}}\int_{||\vv x||>r} \int K_{j,H}(\vv x-\vv u)g(\vv u)\,d\vv u.
\end{align*}
The last integral can be made smaller then an arbitrarily chosen $\eps$ with the appropriate choice of $r$ independently of the index $t$ since, due to Young's convolution inequality,
\[
\int d\vv x\int K_{j,H}(\vv x-\vv u)g(\vv u)\,d\vv u
\]
 is finite. 

To verify the second condition we note first that, due to Lipschitz continuity of the kernel function and the fact that it is defined on a finite closed rectangle $\Delta$, we have 
\begin{align*}
&\int\left\vert f_{j}^{(t)}(\vv x-\vv a)-f_{j}^{(t)}(\vv x)\right\vert\,d\vv x\\
&\leq\int\,d\vv x\int\left\vert(K_{j,H}(\vv x-\vv a-\vv u)-K_{j,H}(\vv x-\vv u)\right\vert\phi(\vv u))\,d\vv u\\
&=\int\phi(\vv u)\,d\vv u\int\left\vert(K_{j,H}(\vv x-\vv a-\vv u)-K_{j,H}(\vv x-\vv u)\right\vert\,d\vv x\\
&\leq {\cal D}||a||
\end{align*}
for some positive constant ${\cal D}$. Hence, for any $\vv a$ such that $||\vv a||<\rho$ we have the integral bounded from above by an expression that does not depend on the function $f_{j}^{(t)}$.  
\end{proof}
Note that the proof of Lemma \eqref{sub_noncom} shows that all of $\lambda_{j}^{(t)}$ are bounded 
uniformly away from zero. Analyzing the proof of Lemma \eqref{lw_semi}, we see that this is enough to guarantee its validity in the case of density functions defined on a non-compact interval. Thus, the combination of Lemmas \eqref{sub_noncom} and \eqref{lw_semi} proves Theorem \eqref{Th_con} once again.

\section{Simulation study}\label{sim_section}
Our simulation study uses the setting of \cite{levine2024smoothed}. In order to make our exposition self-contained, we describe this setting again in brief. In particular, we use $100$ replications of four independent artificial datasets of sizes $n=300$, $500$, $700$, and $900$ that have been generated from the mixture model \eqref{model} with $m=3$ clusters of equal proportions. The densities involved are bivariate with the dependence mechanism described by a Farlie-Gumbel-Morgenstern (FGM) copula \cite{nelsen2007introduction} pp. $77-84$ with parameters $-0.5$, $0.5$, and $0$, respectively. The marginal densities are normal and Laplace distributions with mean $\mu$ and standard deviation $\sigma$. Specific choices for each of the clusters are given in the Table ~\ref{tab:marginals}. 
\begin{table}[h!t]
  \centering
  \begin{tabular}{l|ccc}
    &cluster 1           &cluster 2           &cluster 3\\ \hline
    dim 1&$\mathrm{N}(-3,2^2)$&$\mathrm{N}(0,0.7^2)$&$\mathrm{N}(3,1.4^2)$\\
    dim
    2&$\mathrm{L}(0,0.7^2)$&$\mathrm{L}(3,1.4^2)$&$\mathrm{L}(0,2.8^2)$\\
  \end{tabular}
  \caption{Marginals used for the numerical experiment of
    Section~\ref{sim_section}.} 
  \label{tab:marginals}
\end{table}
To initialize the algorithm, the data are partitioned into $m$ groups using the $k$-means algorithm. The initial weights are set to be equal to proportions of observations belonging to each group. The starting densities $f_{j}^{0}$, $j=1,\ldots,m$ are bivariate Nadaraya-Watson kernel estimates obtained using points assigned to a particular group by the $k$-means algorithm.  If necessary, they are ``clipped" to the height no less than $10^{-3}$ away from zero and renormalized. The starting values of copula parameters $\rho_{j}^{(0)}$, $j=1,\ldots,m$ are obtained using the starting densities $f_{j}^{(0)}$ and the corresponding starting marginal densities $f_{jk}^{(0)}$, $k=1,\ldots,d$ using the same process as in steps \eqref{est_copula} and \eqref{copula_opt} of the algorithm in the Section \eqref{pr_alg}. The bandwidth matrix is selected using the two-stage plug-in selector of \cite{duong2003plug}. This selector is based on the idea of selecting the first stage pilot bandwidth as a minimizer of the SAMSE (Sum of Asymptotic Mean Squared Errors) criterion. The suggested selector is implemented using the $R$ package $ks$ \cite{duong2007ks}. The bandwidth matrices are kept fixed after initialization to preserve monotonicity of the algorithm. 

Note that, in order to prove the convergence of functional sequence $\left\{\bm\lambda^{(t_{l})},\vv f^{(t_{l})}\right\}$ in the previous section, we had to impose some additional assumptions on the copula density function, such as boundedness away from zero and the Lipschitz continuity. It has been our observation in practice that the proposed algorithm seems to perform quite satisfactorily even if not all of these assumptions are satisfied. For example, the FGM copula density used in our simulation setting is only uniformly bounded away from zero when the FGM parameter is equal to zero (in which case it simply corresponds to the product copula); nevertheless, our results don't seem to suffer too much from this violation of assumptions of Section \eqref{Con_func}.
\begin{figure}[h!b]
  \centering
  \includegraphics[width=0.7\textwidth]{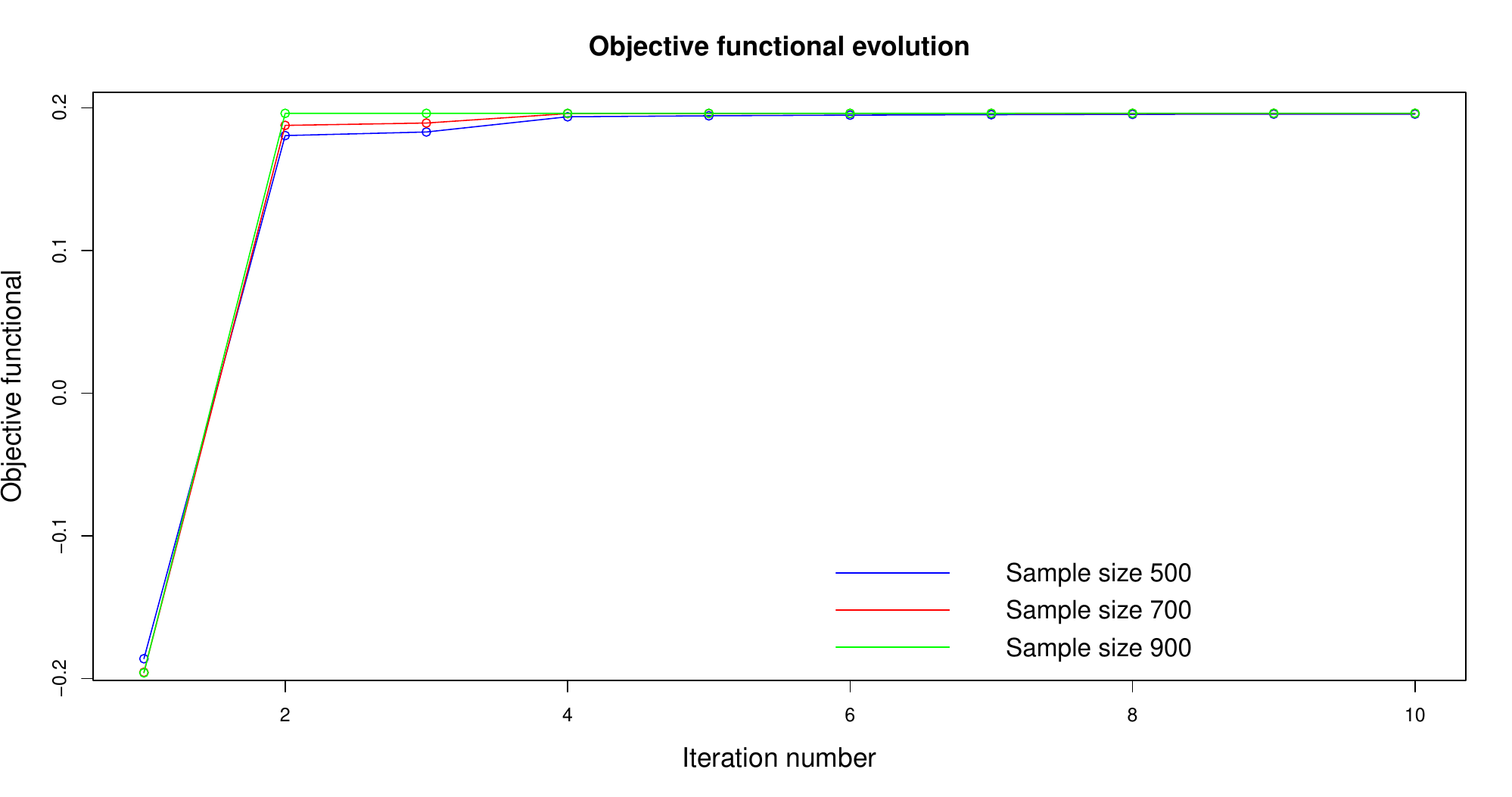}
  \caption{A realization of the evolution process of the sample version of the objective functional over the course of $10$ iterations for three different sample sizes: $500$, $700$, and $900$.}
  \label{Obj.functional}
\end{figure}

The Figure ~\ref{Obj.functional} shows the behavior of the sample version of the objective functional as the number of iterations grow for three different sample sizes $500,700$ and $900$.  Note that the trajectory stabilizes very quickly: after the third or fourth iteration at most, the change in the value of the objective functional is almost imperceptible, even for the sample size $500$.

\begin{figure}[h!t]
  \centering
  \includegraphics[width=0.55\textwidth]{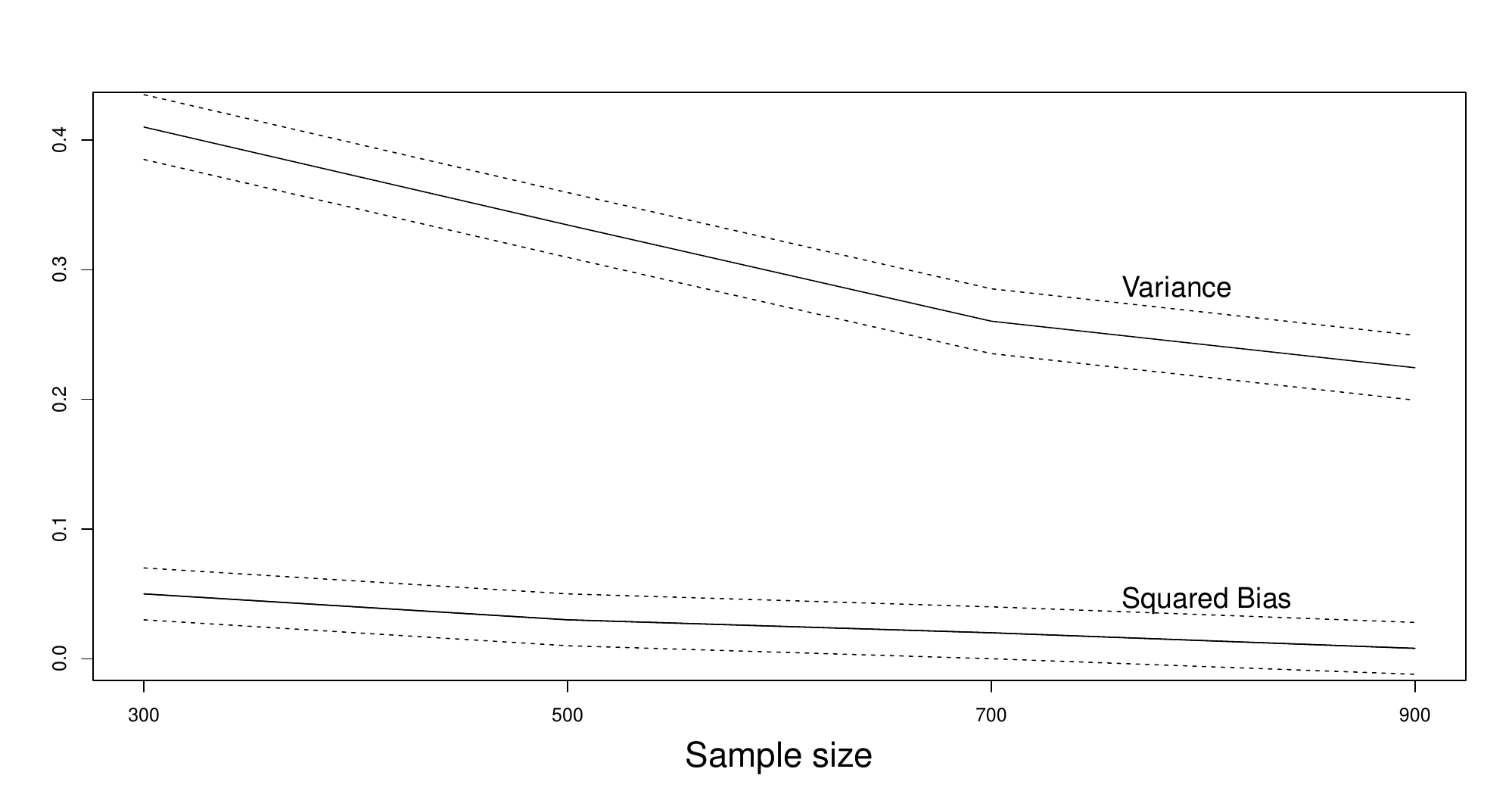}
  \caption{Estimated squared bias and variance of the copula parameter
    vector estimator for various sample sizes at the
    last step of the algorithm. Dashed blue lines represent 95\%
    confidence bands (aka simultaneous confidence intervals) obtained
    from an application of the multivariate central limit theorem to
    the five hundred replications.}
  \label{sq.bias}
\end{figure}
Figure ~\ref{sq.bias} illustrates the behavior of the estimated squared bias and variance of the entire copula parameter vector. Note that, for smaller values of the sample size $n$, the variance is several times higher than the squared bias; however, the variance decays at a fairly fast rate afterwards. More specifically, at $n=900$ the variance is approximately $50\%$ to $70\%$  smaller than the variance at $n=300$. This suggests the convergence rate somewhat slower than the classical parametric rate of $\frac{1}{n}$. 

\begin{figure}[h!]
  \centering
  \caption{True and estimated marginal densities of the three clusters
    and the two dimensions for $n=900$.The marginal estimates are obtained by integrating 
    estimated bivariate densities.}
  \subfloat[]{
    \label{fig:subfig:d1c1}
    \includegraphics[width=.3\textwidth]{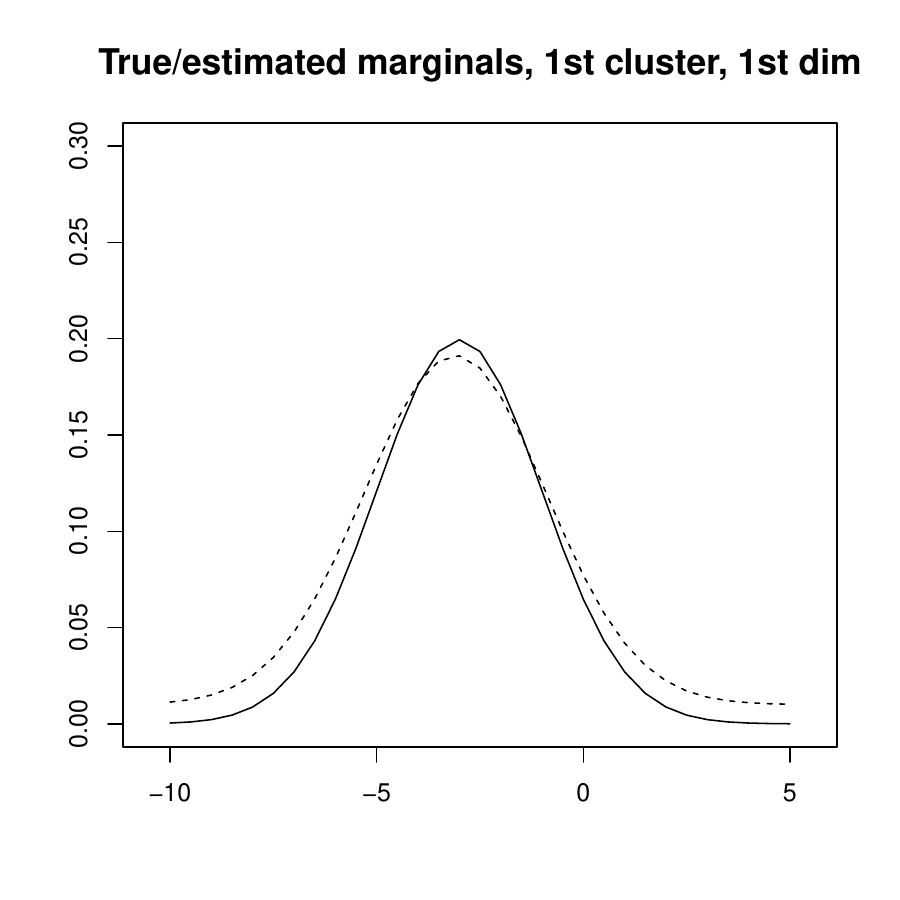 }}
    \hfill
  \subfloat[]{
    \label{fig:subfig:d1c2}
    \includegraphics[width=.3\textwidth]{ 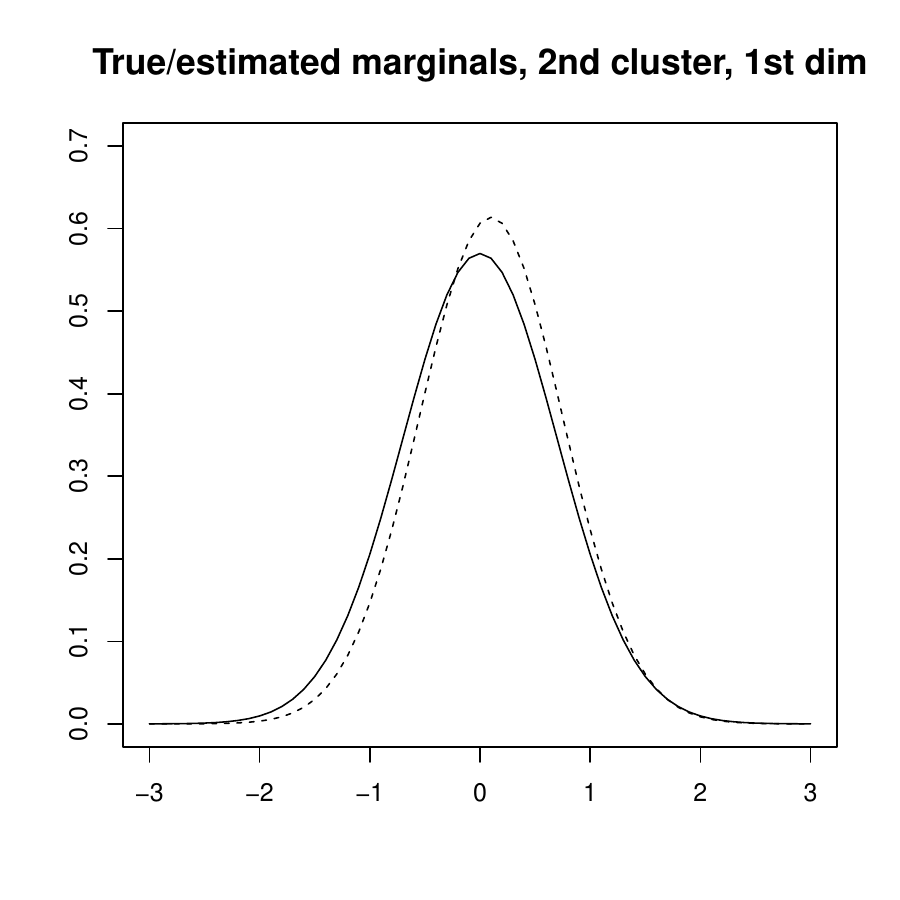}}
    \hfill
  \subfloat[]{
    \label{fig:subfig:d1c3}
    \includegraphics[width=.3\textwidth]{ 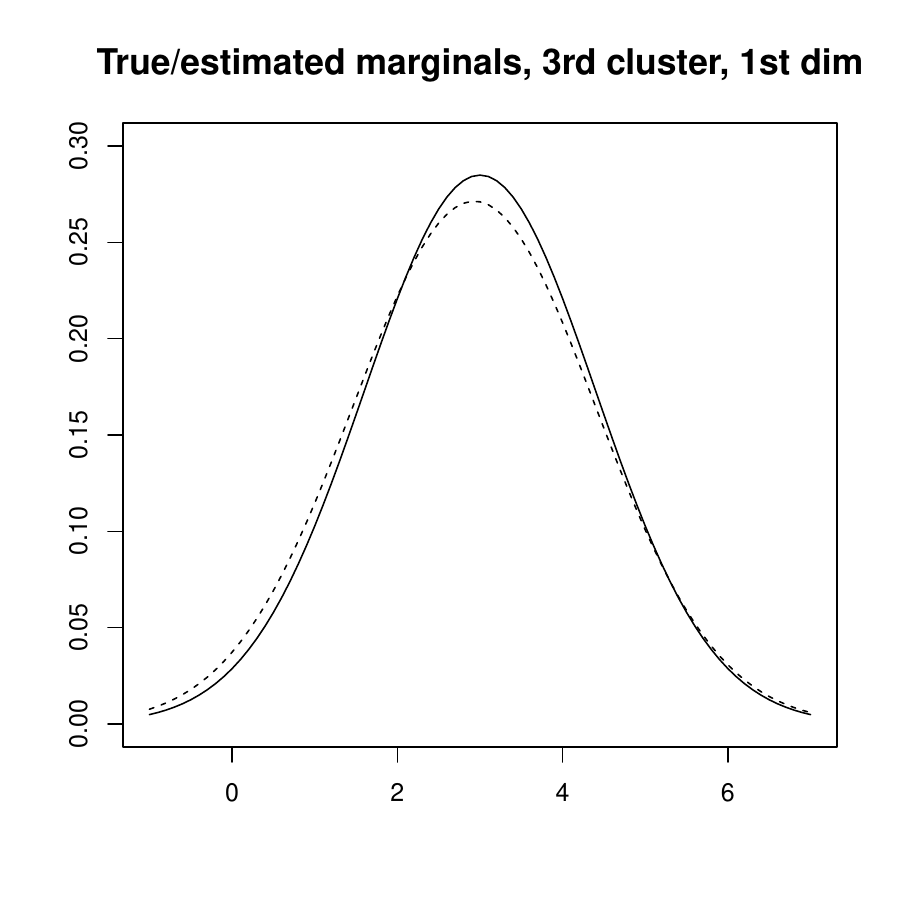} }
    \hfill
  \subfloat[]{
    \label{fig:subfig:d2c1}
   \includegraphics[width=.3\textwidth]{ 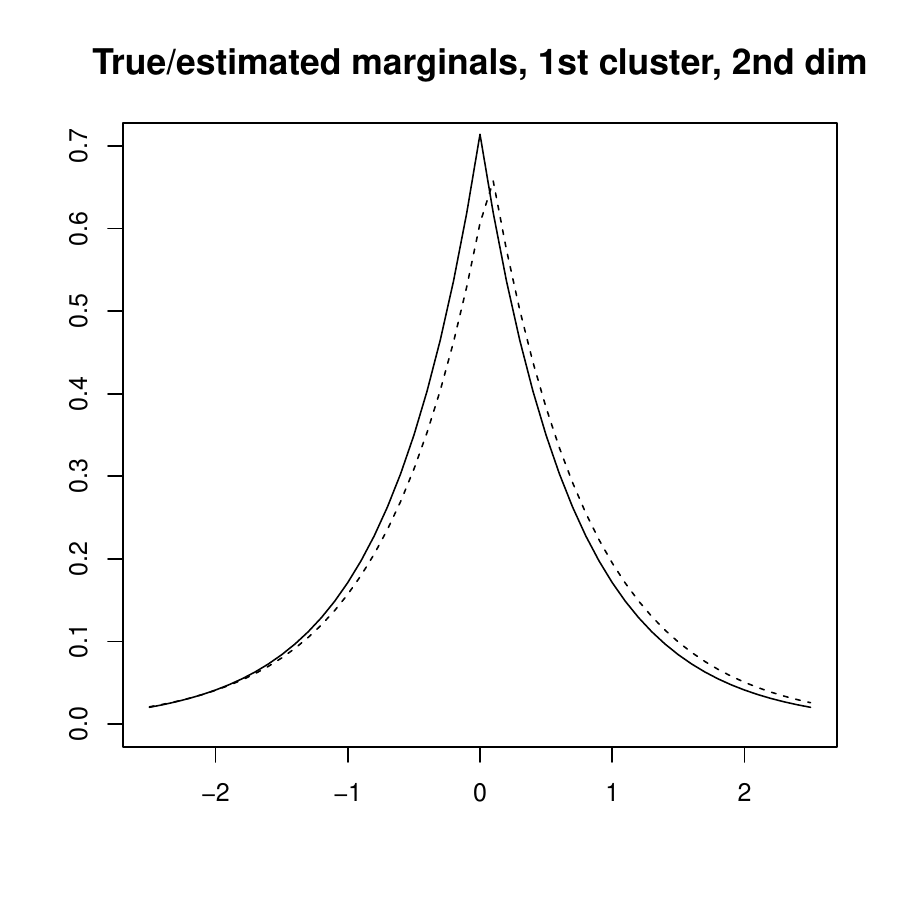} }
   \hfill
   \subfloat[]{
   \label{fig:subfig:d2c2}
   \includegraphics[width=.3\textwidth]{ 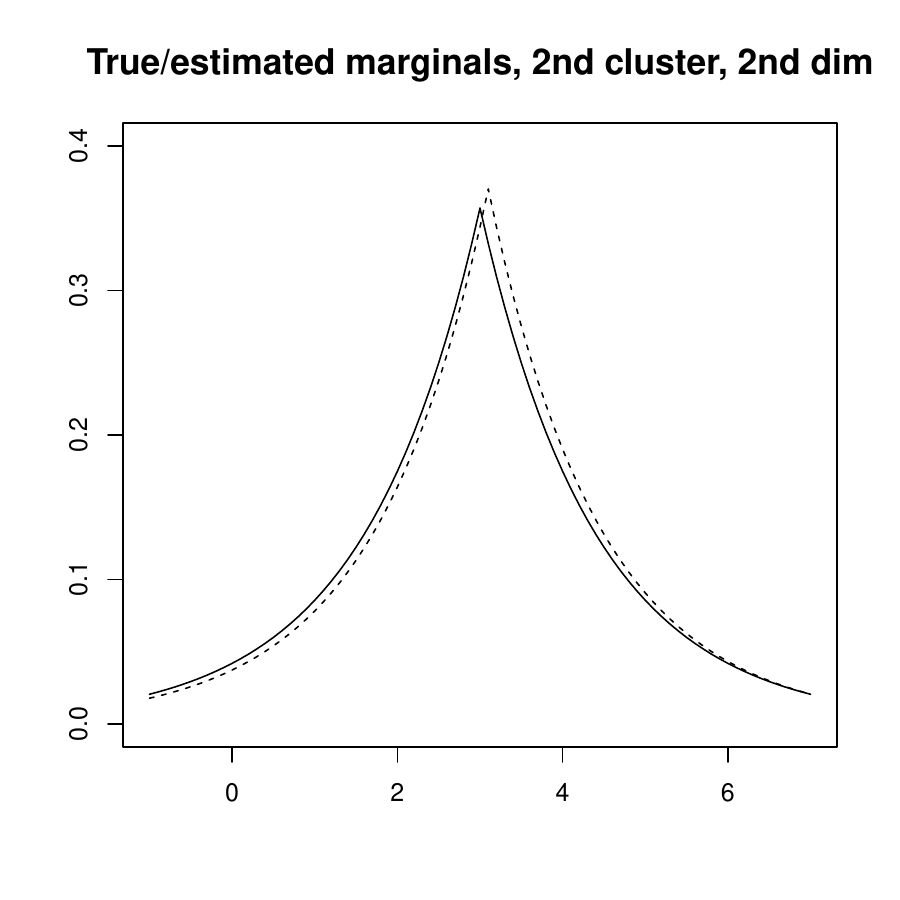} }
   \hfill
   \subfloat[]{
   \label{fig:subfig:d2c3}
   \includegraphics[width=.3\textwidth]{ 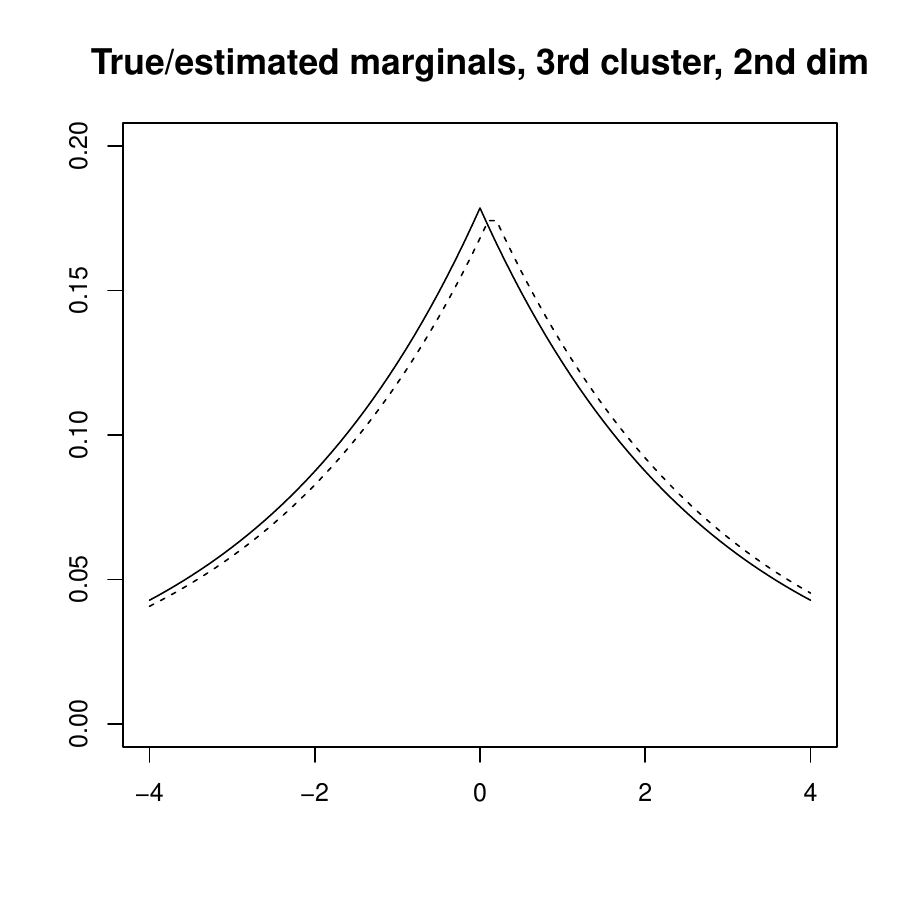} }
  \thisfloatpagestyle{empty}
  \label{marginals_kmeans}
\end{figure}
Figure ~\ref{marginals_kmeans} illustrates estimated marginal densities at the last step of the algorithm for the sample size $n=900$, for the last replication. Each marginal density $f_{jk}(x)$ is superimposed on the estimate $\hat f_{jk}(x)$, obtained at the last step of the algorithm, on a separate plot for $k=1,2$ and $j=1,2,3$. In every one of these plots, a solid line denotes the true density function and the dashed line denotes an estimated one. Note that the proposed algorithm, unlike the algorithm of \cite{levine2024smoothed}, only estimates the multivariate density functions $f_{j}(\vv y)$ directly. The corresponding marginal densities $f_{jk}(y_k)$, $k=1,\ldots,d$ are obtained using numerical marginal integration of the estimated multivariate densities. One can expect this to produce somewhat less precise estimates of the marginal densities than the algorithm of \cite{levine2024smoothed}. To our surprise, the resulting estimates gave a rather good fit to the true densities. To perform our calculations, we used R version 4.4.2 on a Dell laptop with $32$ GB RAM and Intel(R) Core(TM) Ultra 7 165U processor. Every iteration of our algorithm took about $1$ minute with the sample size $n=900$. 

\begin{figure}[h]
  \centering
  \caption{True and estimated marginal densities of the three clusters
    and the two dimensions for $n=500$. The marginal estimates are obtained by integrating 
    estimated bivariate densities.}
  \subfloat[]{
    \label{fig:subfig:d1c1500}
    \includegraphics[width=.3\textwidth]{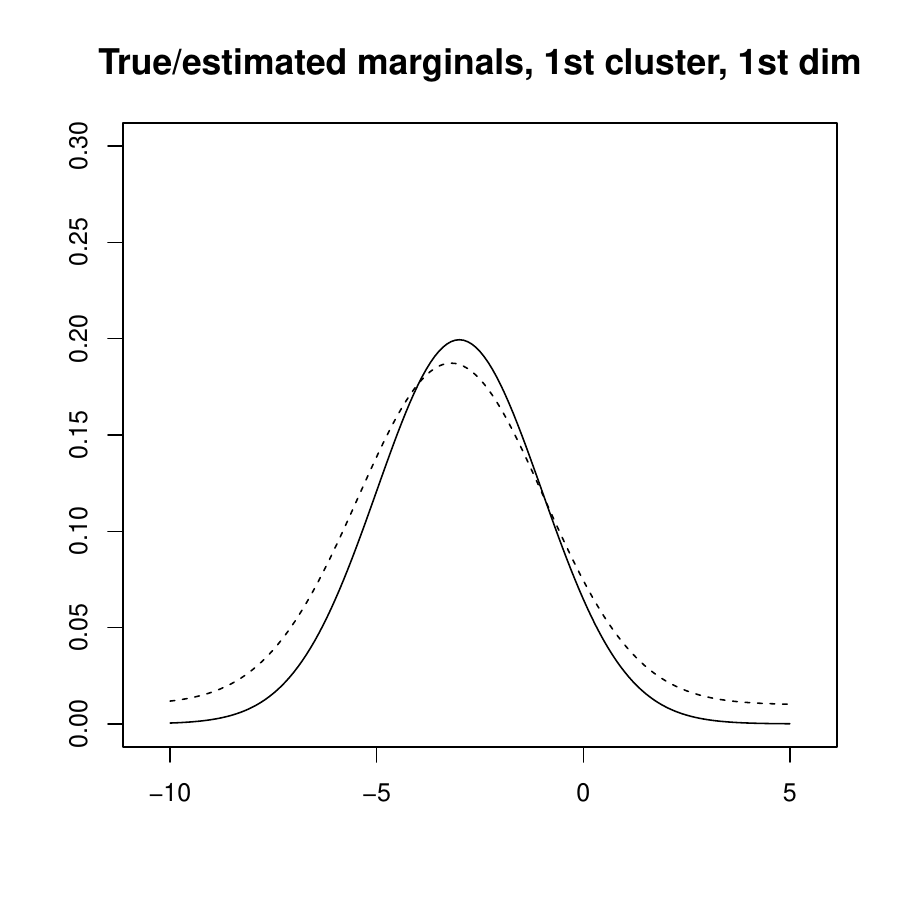 }}
    \hfill
  \subfloat[]{
    \label{fig:subfig:d1c2500}
    \includegraphics[width=.3\textwidth]{ 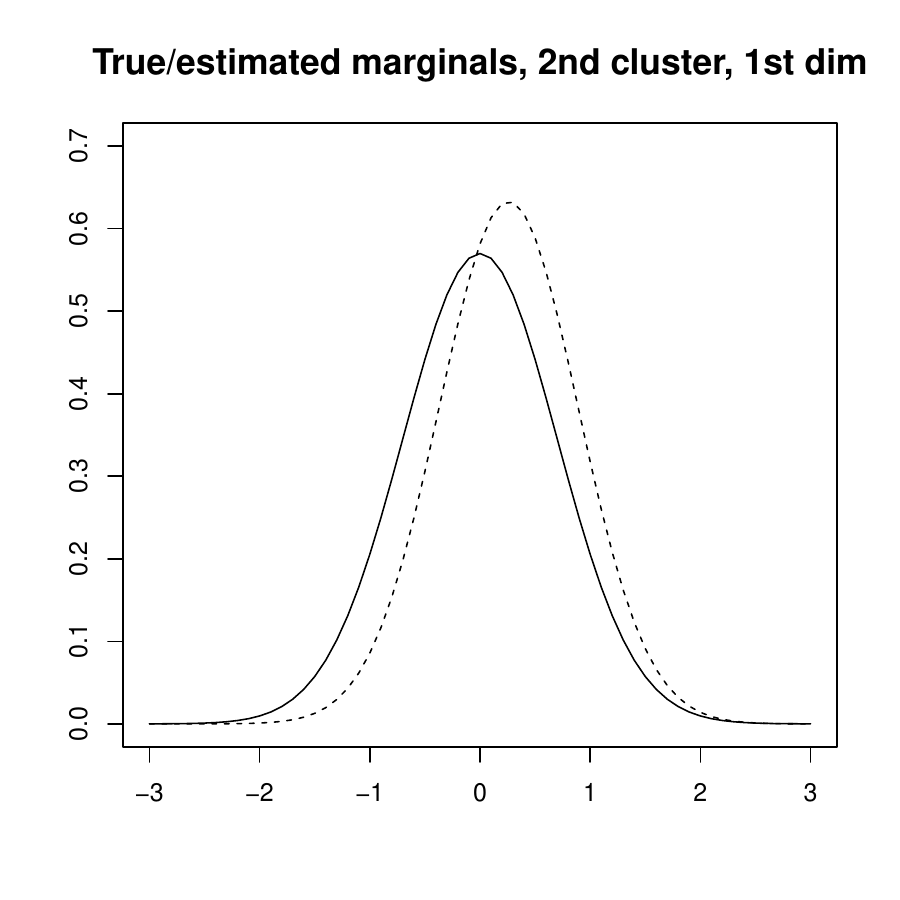}} 
    \hfill
  \subfloat[]{
    \label{fig:subfig:d1c3500}
    \includegraphics[width=.3\textwidth]{ 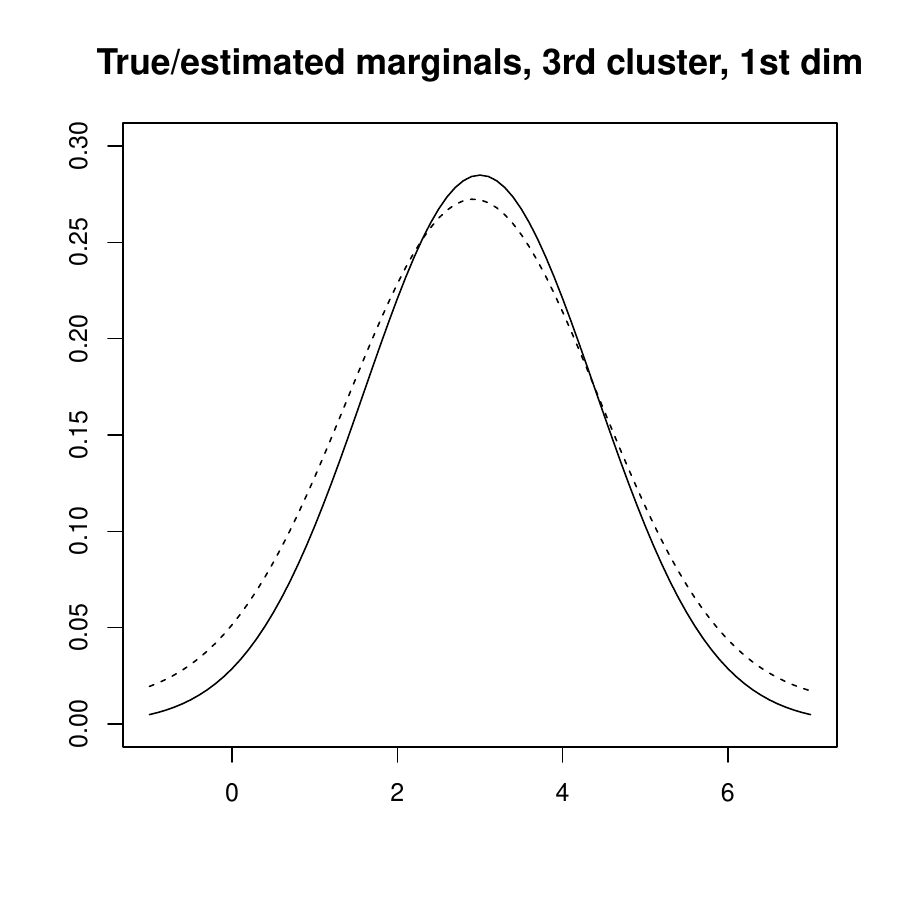} }
    \hfill
  \subfloat[]{
    \label{fig:subfig:d2c1500}
   \includegraphics[width=.3\textwidth]{ 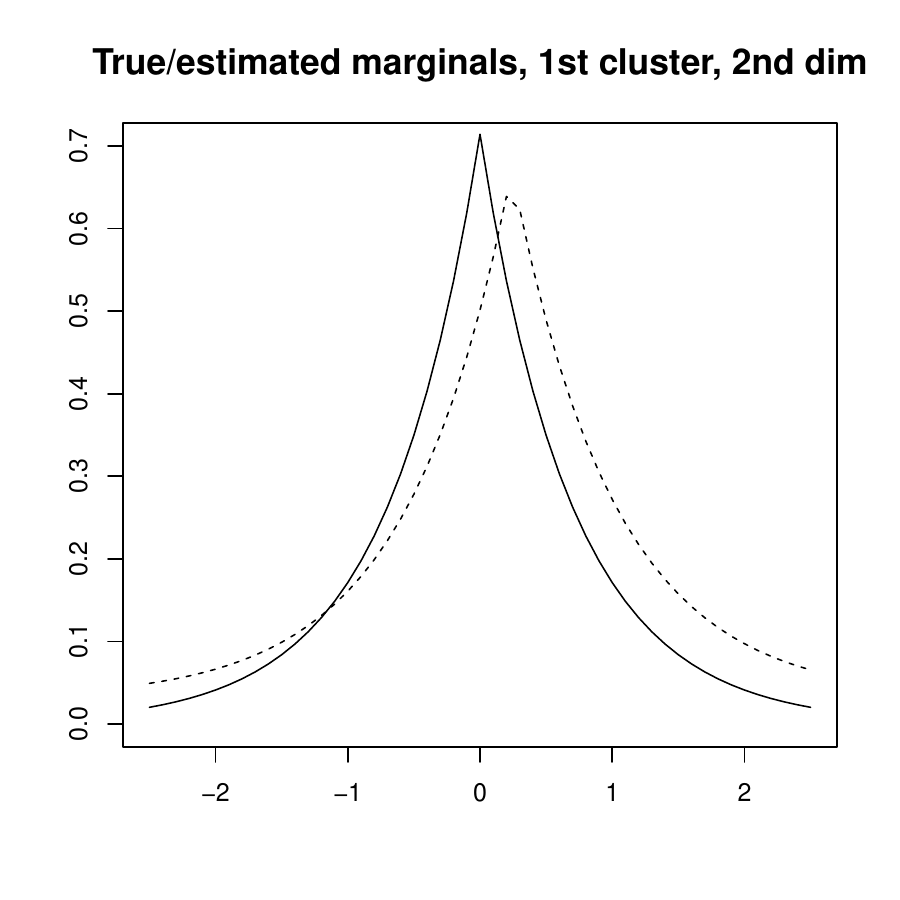} }
   \hfill
   \subfloat[]{
   \label{fig:subfig:d2c2500}
   \includegraphics[width=.3\textwidth]{ 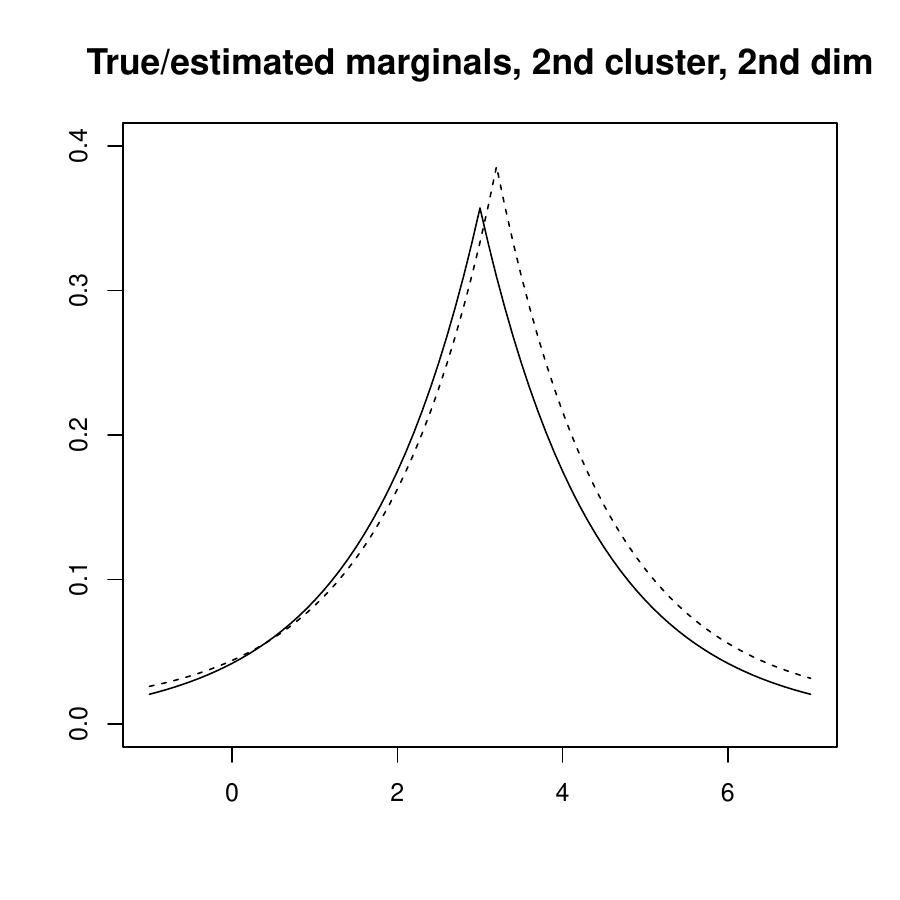} }
   \hfill
   \subfloat[]{
   \label{fig:subfig:d2c3500}
   \includegraphics[width=.3\textwidth]{ 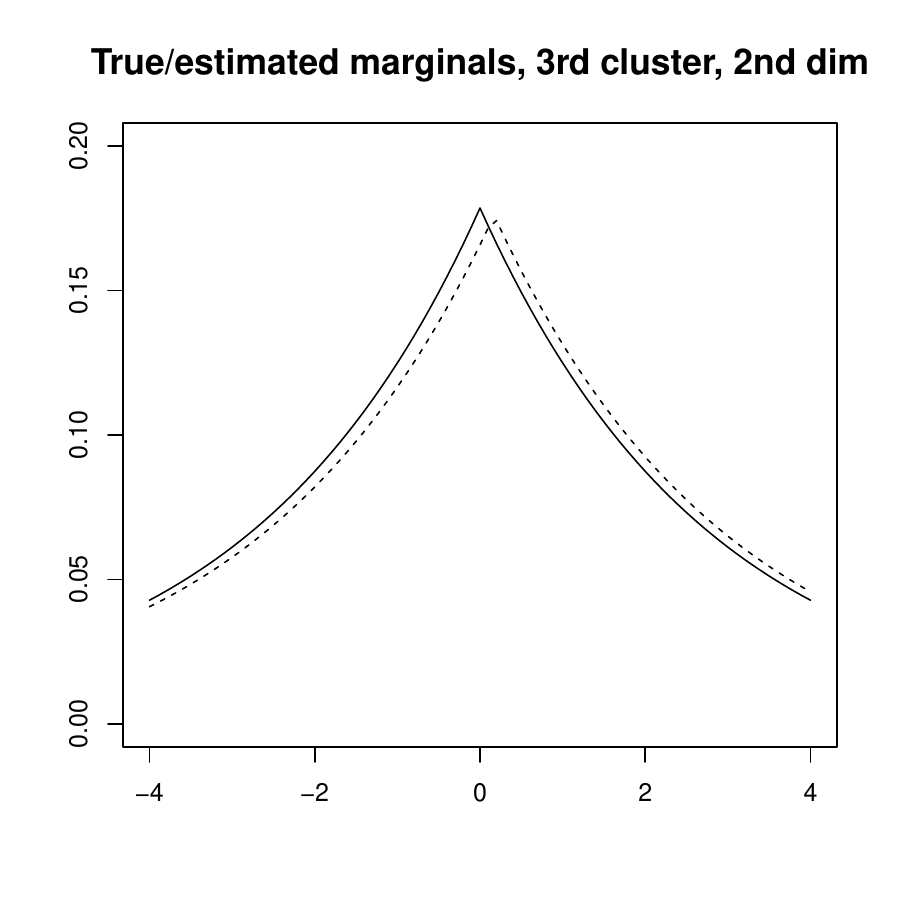} }
   \thisfloatpagestyle{empty}   
  \label{marginals_kmeans_500}
\end{figure}
We have also included another Figure ~\ref{marginals_kmeans_500} to illustrate how well the marginal densities can be estimated at a lower sample size $n=500$. Note that the performance of our algorithm has now deteriorated somewhat. All of the marginal density estimates show noticeable bias compared to the case $n=900$. Moreover, the estimate of the Laplace density in the second dimension of the first component show inability to fit the cusp of the distribution correctly, presenting what seems to be a smoothing artefact due to a smaller sample size. 

As is the case for all the algorithms of this type, its convergence is of the local nature. Due to this, we also conduct a brief investigation to assess the impact of initialization on the results produced by the proposed algorithm. To do so, we use the last dataset generated in the previous simulation experiment with $n=900$. Initialization of the algorithm is now performed using a Gaussian mixture model with independent components. In other words, in the first step of the algorithm the marginal vector $\vv f^{0}$ and the vector of probability weights $\vv \lambda^{0}$ are obtained as estimates of components and weights of a Gaussian mixture model with independent marginals. The results are shown in the Figure ~\ref{marginals_kmeans_900_G}. 
\begin{figure}[h]
  \centering
  \caption{True and estimated marginal densities of the three clusters
    and the two dimensions for $n=900$. The marginal estimates are obtained by integrating 
    estimated bivariate densities. The initialization using generalized Gaussian mixtures is used.}
  \subfloat[]{
    \label{fig:subfig:d1c1900_G}
    \includegraphics[width=.3\textwidth]{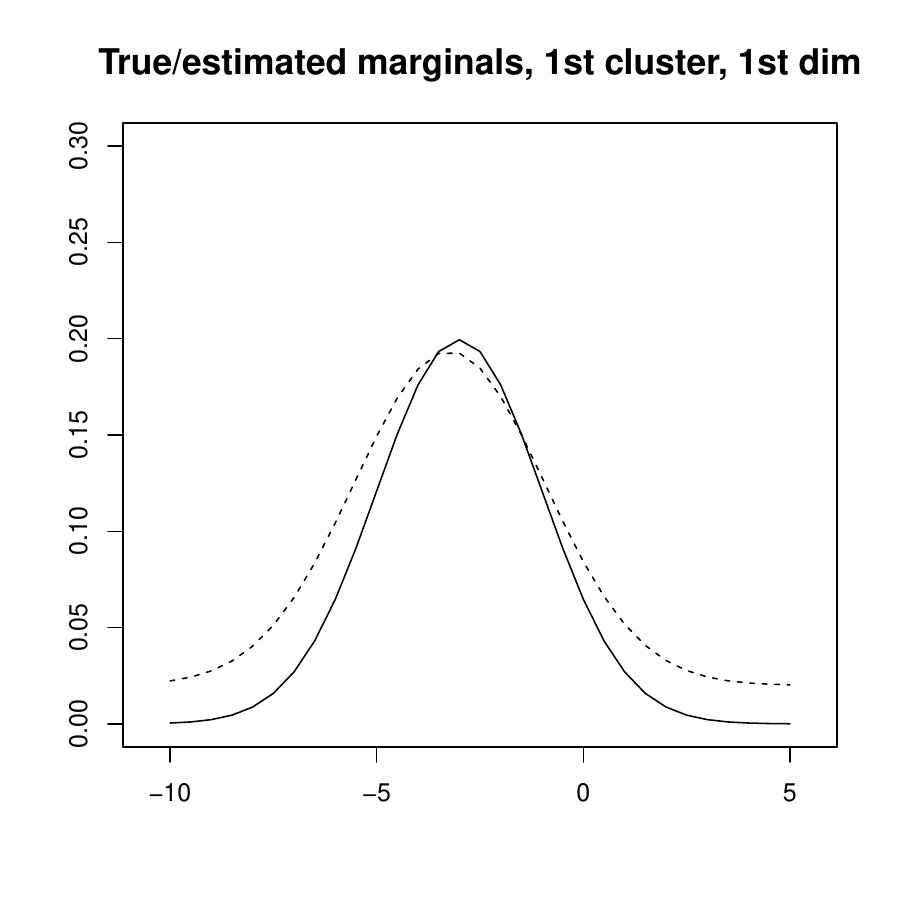 }}
    \hfill
  \subfloat[]{
    \label{fig:subfig:d1c2900_G}
    \includegraphics[width=.3\textwidth]{ 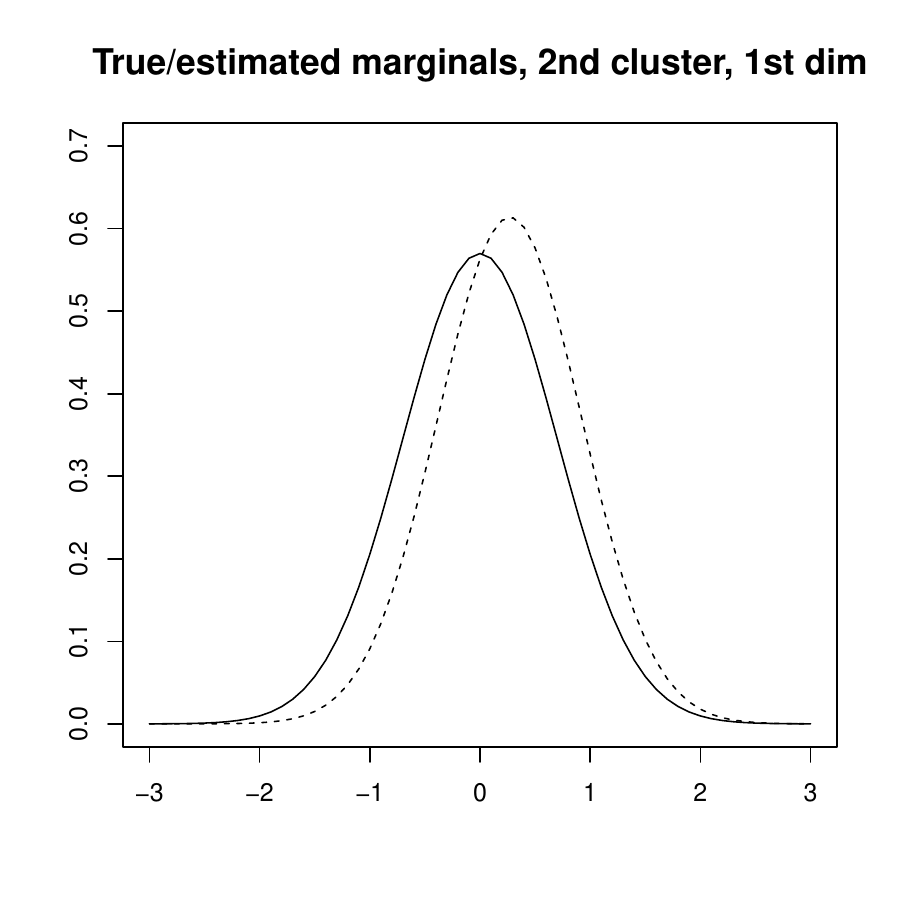} }
    \hfill
  \subfloat[]{
    \label{fig:subfig:d1c3900_G}
    \includegraphics[width=.3\textwidth]{ 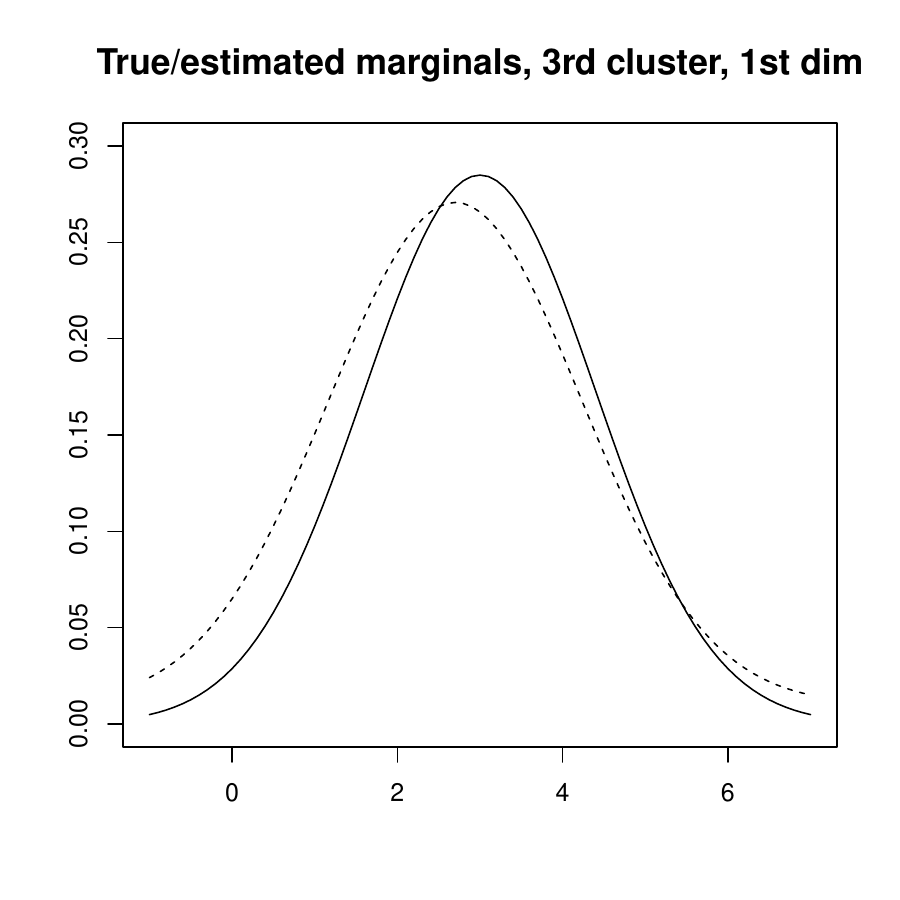} }
    \hfill
  \subfloat[]{
    \label{fig:subfig:d2c1900_G}
   \includegraphics[width=.3\textwidth]{ 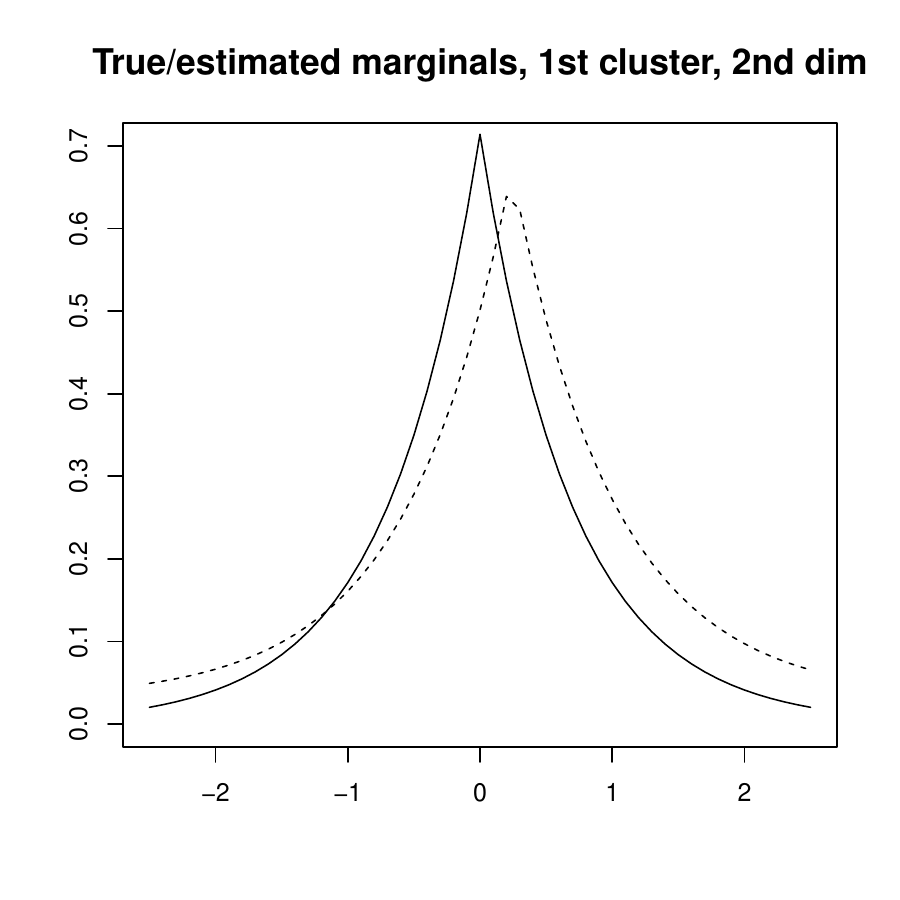} }
   \hfill
   \subfloat[]{
   \label{fig:subfig:d2c2900_G}
   \includegraphics[width=.3\textwidth]{ 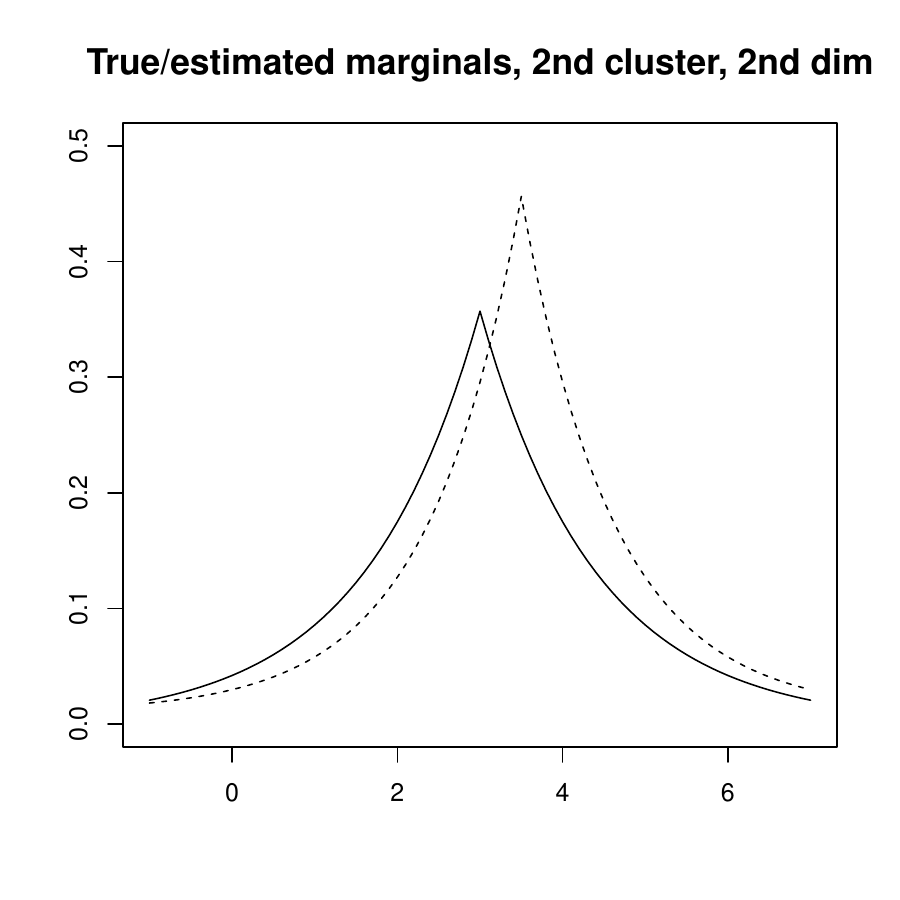} }
   \hfill
   \subfloat[]{
   \label{fig:subfig:d2c3900_G}
   \includegraphics[width=.3\textwidth]{ 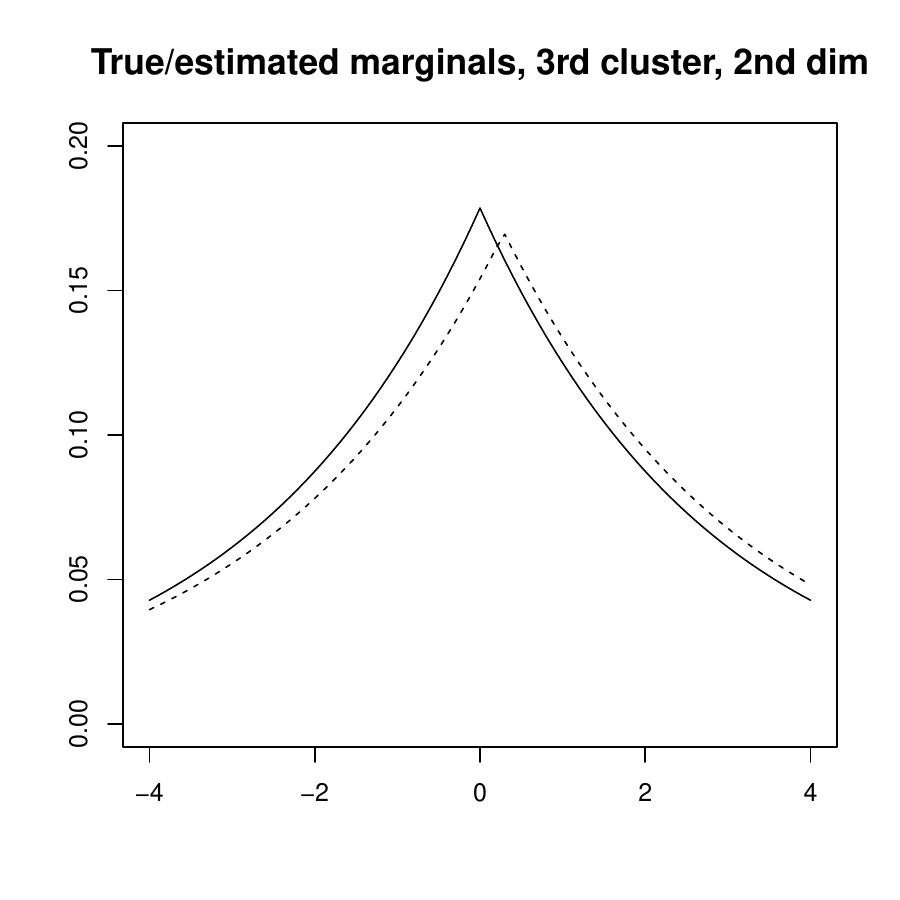} }
  \label{marginals_kmeans_900_G}
 \thisfloatpagestyle{empty} 
\end{figure}
In this case, the quality of estimates of marginal densities has deteriorated considerably; in particular, it is especially noticeable when looking at estimates of the first marginal of the first and third components, respectively, and comparing them to corresponding plots in Figures ~\ref{marginals_kmeans} and ~\ref{marginals_kmeans_500}. This suggests that initialization does play a role in the performance of our algorithm. This is also confirmed by the results of the study of copula parameter estimation. Figure ~\ref{copula.parameters} shows the evolution of values of the three components of the estimated copula parameter vector as the number of iterations increases. On the right-hand side, we see that, with the Gaussian mixture initialization, two out of the three parameters seem to be converging to the wrong limits of $0.2$ and $-0.2$ respectively; only one of the estimates seem to be stable at around the right limit of $-0.5$. On the contrary, the $k$-means initialized process shows a much better behavior, with all three estimates converging to limits that seem approximately correct. 

\begin{figure}[h]
  \centering
  \subfloat[]{\label{fig:coppar-across-n900:subfig:GMM}
    \includegraphics[width=.5\textwidth]{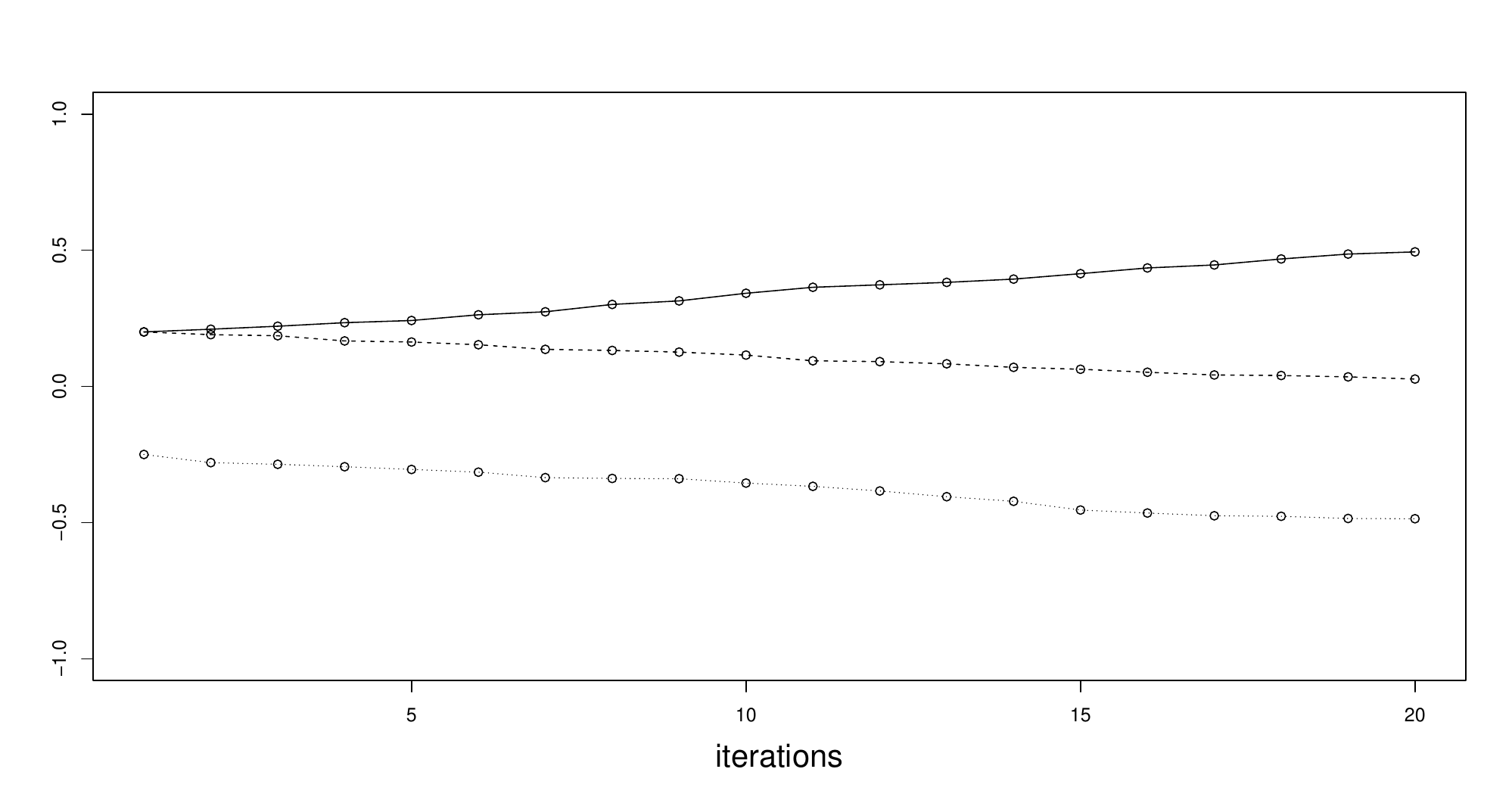}}
  \subfloat[]{\label{fig:coppar-across-n900:subfig:KNN}
    \includegraphics[width=.5\textwidth]{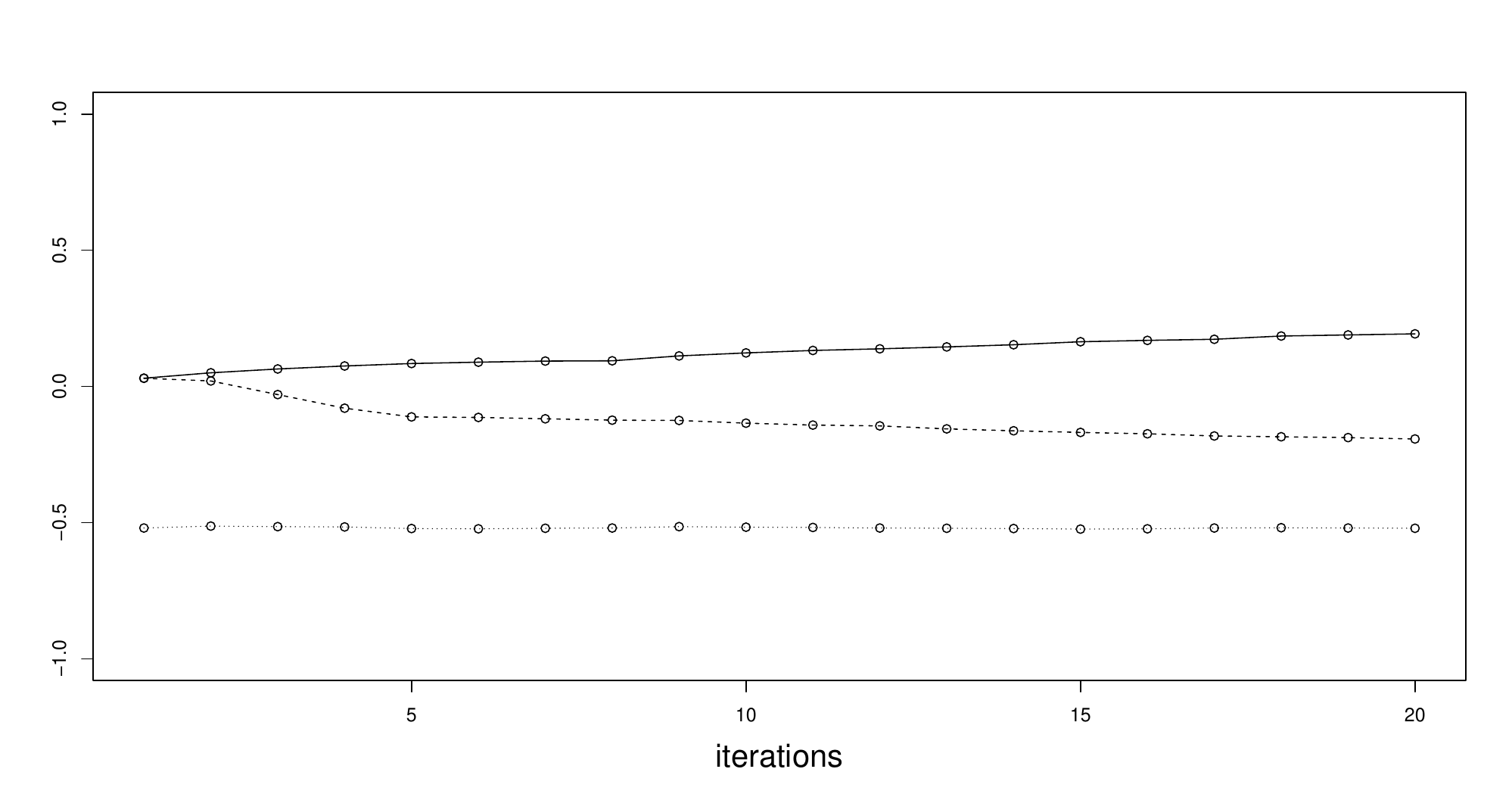}}
  \caption{Values of the three components of the estimated copula
    parameter vector across iterations for the last dataset generated
    in Section~\ref{sim_section} with $n=900$ and initialization by fitting
    \protect\subref{fig:coppar-across-n900:subfig:GMM} a $k$-means algorithm and \protect\subref{fig:coppar-across-n900:subfig:KNN}
    a Gaussian mixture model. Solid, dashed and dotted lines correspond to estimators of the true parameter equal to $0.5$, $0$ and $-0.5$, respectively}
  \label{copula.parameters}
\end{figure}

Note that the proposed algorithm has approximately the same relationship with the algorithm algorithm of \cite{levine2024smoothed} as the monotonic npMSL algorithm of \cite{Biometrika_2011LevineHunterChauveau} to the earlier algorithm of \cite{benaglia2009like}.  In other words, it proposes a monotonic (with respect to a specific objective functional) version of the algorithm \cite{levine2024smoothed}. This similarity also extends to a different aspect of the relationship between the current algorithm and that of \cite{levine2024smoothed}. Specifically, in our simulations, neither algorithm showed itself to be consistently superior to the other - exactly the same situation as with \cite{Biometrika_2011LevineHunterChauveau} and \cite{benaglia2009like}. Due to this, and in order to be concise, we do not include these comparisons in our manuscript. 
\section{Real data analysis}

We will illustrate the behavior of our algorithm using the famous iris dataset. As a reminder, it has $n=150$ observations of $d=4$ characteristics of flowers of three species from the genus {\it Iris: Iris setosa, Iris versicolor, Iris virginica }. For simplicity, and following the precedent of \cite{levine2024smoothed}, we will only use two of these characteristics as variables: the sepal length and the petal length. The algorithm proposed in Section ~\ref{algo} is used to perform clustering of this dataset with Gaussian copulas. Also, the stopping criterion described in Section ~\ref{algo} is used. The algorithm is initialized using k-means approach and the fixed bandwidth matrix is used throughout the algorithm.  The reason for this is that, as is the case with the algorithm of \cite{Biometrika_2011LevineHunterChauveau}, the bandwidth matrix has to be kept constant in order for the algorithm to retain the desirable monotonicity property. The bandwidth matrix is selected using a two-stage approach that, at its first stage, uses the choice of the matrix that minimizes the so-called SAMSE (Sum of the Asymptotic Mean Squared Errors) criterion proposed by \cite{duong2003plug}. This approach allows us to avoid selecting diagonal bandwidth matrices only and guarantees that the selected bandwidth matrix is a positive definite one. Selecting the number of clusters in nonparametric density mixtures is a difficult and mostly unsolved problem. Limited results available have only been obtained for models with conditionally independent marginals e.g. \cite{kwon2021estimation} and references therein. Due to this, we only consider the case with the number of clusters equal to the true number of $3$. 

The obtained classification results are reported in Figure ~\ref{fig:iris-classif-clust3:overall}. Note that clustering results obtained using our method are quite close to the true classification and are noticeably better than those from the Gaussian mixture method. Only three observation points end up being misclassified in the {\it Iris setosa} cluster at the lower left side of the plot. These points should have been in the {\it Iris versicolor} class. Our results also compare quite favorably to those of \cite{zhu2019clustering}. They use a rather different algorithm that is based on applying the independent component analysis (ICA) transformation to the original data to make the coordinates as close to being independent as possible as a first step. When that algorithm is applied to the Iris data set in \cite{zhu2019clustering}, seven points out of $150$ end up being misclassified as opposed to the three misclassified points produced by our algorithm. It needs to be also mentioned here that that \cite{zhu2019clustering} use all four variables from the iris dataset for classification purposes whereas our analysis is based on the two variables only.  

\begin{figure}[h!t]
  \centering
  \subfloat[]{\label{fig:iris-classif-clust3}
    \includegraphics[width=.5\textwidth]{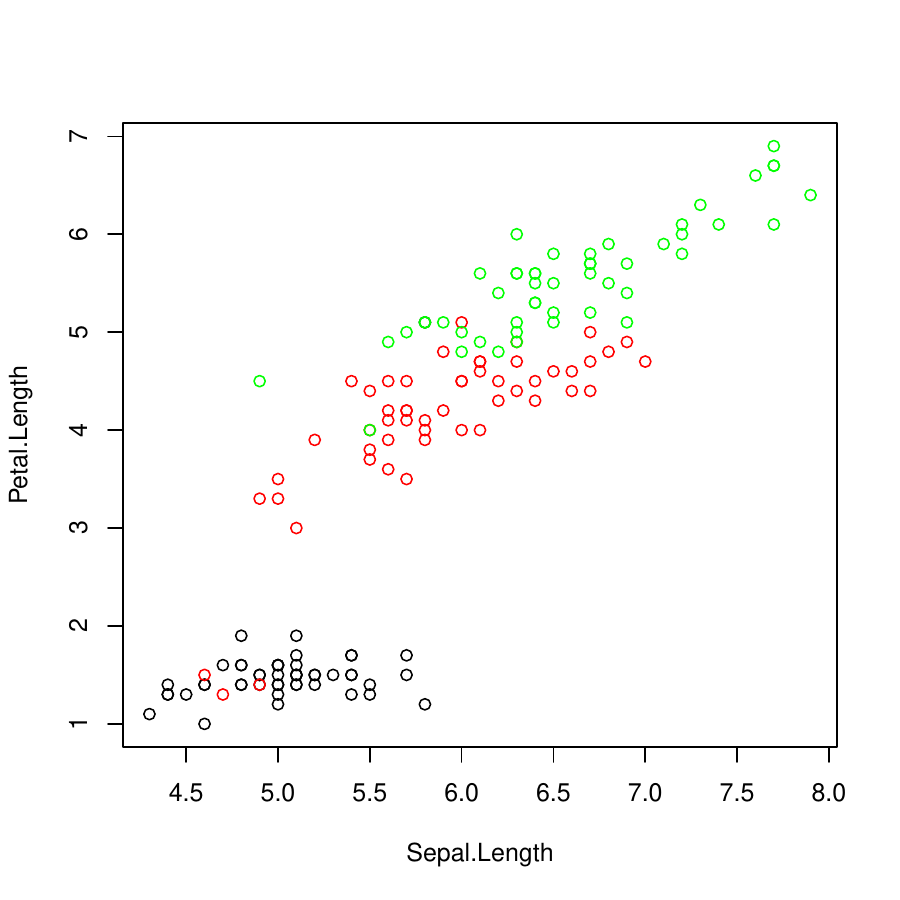}}
  \subfloat[]{\label{fig:iris-classif-gmm3}
  \includegraphics[width=.5\textwidth]{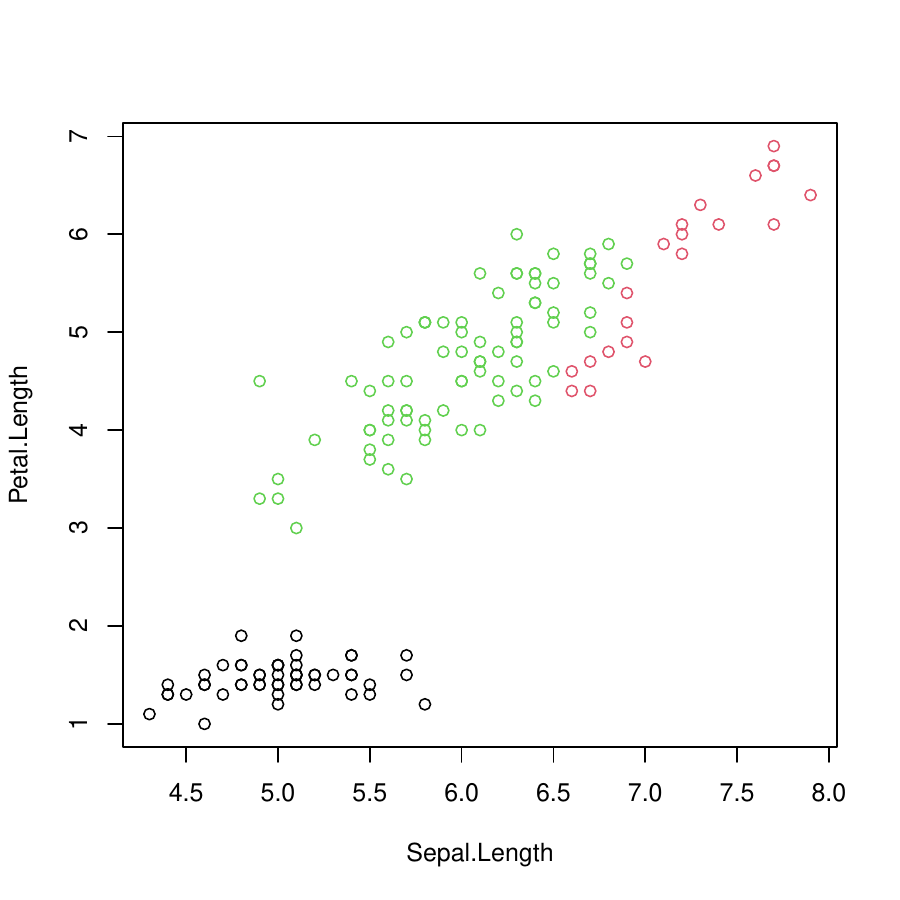}}\\
  \subfloat[]{\label{fig:iris-classif-truth:threeclusters}
  \includegraphics[width=.5\textwidth]{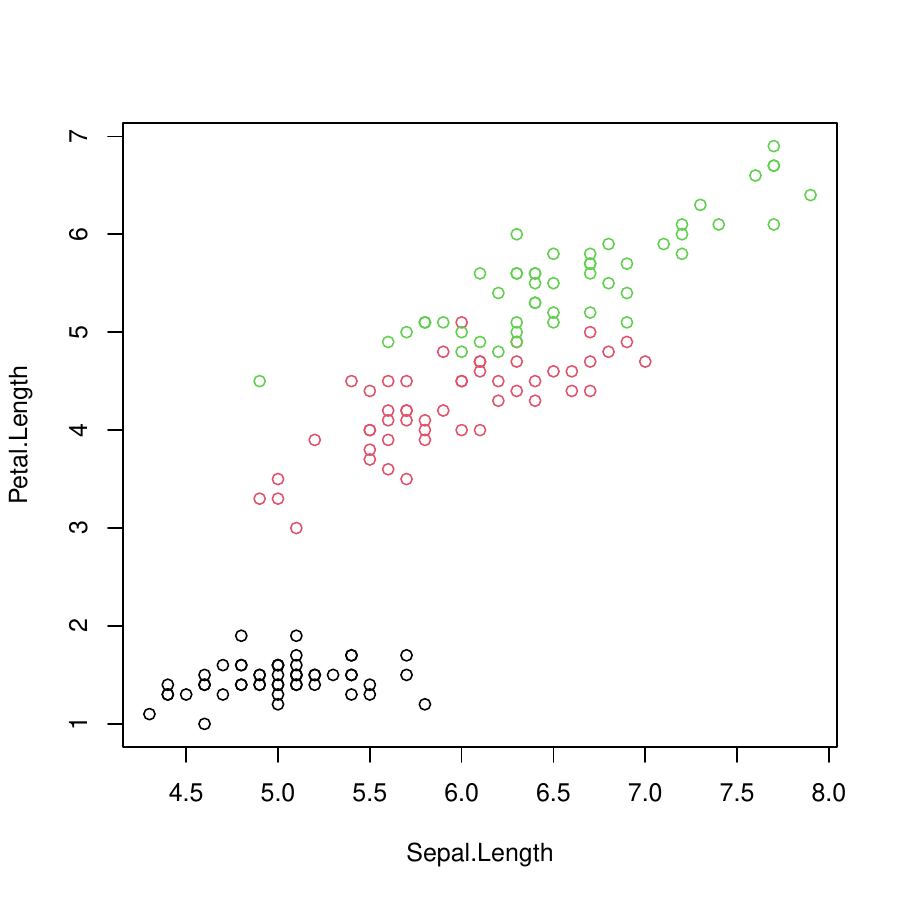}}
  \caption{Iris data: classification 
    based on the choice of 3 clusters. Top left: results of algorithm of
    Section~\ref{algo} with 3 clusters. Top right: results for the Gaussian
    mixture model with 3 clusters. Bottom left: true classification.}
  \label{fig:iris-classif-clust3:overall}
\end{figure}

\section{Conclusion}

The algorithm proposed in this manuscript is designed to estimate the parameters of an arbitrary copula-based semiparametric finite density mixture model. It has been developed to bypass an obvious limitation of the earlier algorithm proposed in \cite{levine2024smoothed} - its lack of monotonicity. The model considered is a general one - the marginal densities are not viewed as belonging to any specific family and the dependence between them is modeled using a copula. The algorithm is a deterministic algorithm of the MM type that shows good performance in illustrative numerical examples. 

As is common in statistical research, the problem considered and the algorithm suggested to solve it raise a number of additional questions worthy of consideration. First, to the best of our knowledge, nothing is known about the identifiability of the model \eqref{model}-\eqref{cmpt} considered in this manuscript. Moreover, even more circumscribed model considered in \cite{mazo2019constraining} where all of the marginal density functions are assumed to belong to a location-scale family is not known to be identifiable either. Investigating identifiability of these models represents an important avenue for future research. 

Another important topic of future research would be to make this algorithm feasible for higher dimensional data. We believe that the starting point of this future research would be an attempt to model dependence of coordinates of multivariate distributions using Archimedean copulas \cite{nelsen2007introduction} pp. $109-155$. There are two reasons for this, in our opinion. The first one is that Archimedean copulas can be written down in closed form which makes optimization considerably easier. Second, they can model a fairly wide variety of dependence structures using only one- or two-dimensional parameters.  

As has been pointed out earlier in e.g. \cite{Biometrika_2011LevineHunterChauveau}, there is a certain sense in updating the bandwidth matrix at every step of iteration since this takes into account updated estimates of component densities. The same remark is also applicable to updating values of copula parameters $\rho_{j}$, $j=1,\ldots,m$ used to define the smoothing kernel. At the same time, such an approach would violate the monotonicity of the proposed algorithm. A possible research question would be to quantify the resulting loss of monotonicity with respect to the objective functional considered in this manuscript and, moreover, to investigate just how substantial such a violation of monotonicity is depending on, for example, the specific mechanism of the bandwidth matrix selection. The most pressing question would be, then, to show if the algorithm still converges despite the loss of monotonicity property. 

%\section*{Acknowledgements}

%Michael Levine's research has been partially funded by the National Science Foundation grant \# $2311103$.

%\section*{Declaration}

%The author has no relevant financial or non-financial interests to disclose. 

\clearpage
\bibliographystyle{plainnat} 
\bibliography{reference}   
\end{document}